\documentclass[11pt,a4paper]{article}

\usepackage{amsmath}
\usepackage{amsthm}
\usepackage{amssymb}
\usepackage{stmaryrd}
\usepackage{bbold}

\usepackage{graphicx} % Package to include images
\usepackage{grffile}

\usepackage{booktabs}
\usepackage{tabularray}
\usepackage{multirow}
\usepackage{subcaption}

\usepackage{bm}
\usepackage{ulem}
\usepackage[dvipsnames]{xcolor}

\usepackage{listings}

\usepackage[square,sort,comma,numbers]{natbib}

\usepackage{algorithm}
\usepackage{algorithmic}

\usepackage[utf8]{inputenc}
\usepackage{dsfont}

\usepackage{nicematrix}

\usepackage{xargs}

\newtheorem{lemma}{Lemma}
\newtheorem{proposition}{Proposition}
\newtheorem{corollary}{Corollary}
\newtheorem{remark}{Remark}
\usepackage[pdftex,
bookmarks, bookmarksnumbered=true, bookmarksopen=true,
colorlinks, citecolor=ForestGreen, filecolor=red, linkcolor=red, urlcolor=ForestGreen,
pdftoolbar=true, pdfmenubar=true, pdffitwindow = true
]{hyperref}

\usepackage{url}

\newcommandx*\xvec[2][1=0,2=n]{x_{#1},\ldots,x_{#2}}

\newcommandx*\m[1]{\mathcal{#1}}
\newcommand{\R}{\mathbb{R}}
\newcommand{\Z}{\mathbb{Z}}

\newcommand{\disjcup}{\bigsqcup}%disjoint union

\newcommand{\bx}{\bm{x}}
\newcommand{\bX}{\bm{X}}
\newcommand{\bY}{\bm{Y}}

\newcommand{\ba}{\bm{a}}
\newcommand{\bA}{\bm{A}}
\newcommand{\bb}{\bm{b}}
\newcommand{\bB}{\bm{B}}

\newcommand{\be}{\bm{e}}

\newcommand{\bK}{\bm{K}}

\newcommand{\bt}{\bm{t}}

\newcommand{\bE}{\bm{E}}
\newcommand{\bI}{\bm{I}}

\newcommand{\bmu}{\boldsymbol{\mu}}
\newcommand{\bxi}{\boldsymbol{\xi}}
\newcommand{\bepsilon}{\boldsymbol{\epsilon}}
\newcommand{\bPhi}{\boldsymbol{\Phi}}
\newcommand{\bPsi}{\boldsymbol{\Psi}}
\newcommand{\bSigma}{\boldsymbol{\Sigma}}
\newcommand{\beeta}{\boldsymbol{\eta}}

\newcommand{\bLambda}{\boldsymbol{\Lambda}}

\newcommand{\condmu}{\bmu}%Conditional mean of a Gaussian vector
\newcommand{\condK}{\bSigma}%Conditional covariance of Gaussian vector
\newcommand{\MatKtheta}{\bK_\theta}%Covariance depending on parameter theta
\newcommand{\complexity}{\mathcal{O}}%Complexity of an operation

\newcommand{\blockset}{\mathcal{B}}%Block indices
\newcommand{\sizeblock}{B}%number of blocks
\newcommand{\partition}{\mathcal{P}}%Partition 
\newcommand{\subpartition}{\mathcal{P}}%Sub-partition
    
\newcommand{\subdivset}{\mathcal{S}}%Subdivisions
\newcommand{\setL}{\mathcal{L}}%multi-set indices 
\newcommand{\setLSP}{\setL^{\subdivset}_{\subpartition}}%Multi-set indices depending on subdivisions and partition
\newcommand{\setLSPstar}{\setL^{\subdivset^\star}_{\subpartition^\star}}%Multi-set indices depending on subdivisions and partition
\newcommand{\setLSBj}{\setL^{\subdivset}_{\blockset_j}}%Multi-set indices for one block depending on subdivisions and partition
\newcommand{\setLSBjstar}{\setL^{\subdivset^\star}_{\blockset^\star_j}}%Multi-set indices for one block depending on subdivisions and partition
\newcommand{\ul}{\underline{\ell}}%Elements in setLSP 
\newcommand{\uk}{\underline{k}}%Elements in setLSP
\newcommand{\basis}{\beta}%Basis 
\newcommand{\basisstar}{\basis^{\subdivset^\star}_{\subpartition^\star}}%Update basis
\newcommand{\basisSP}{\basis^{\subdivset}_{\subpartition}}%Basis depending on subdivisions and partition
\newcommand{\basisSPstar}{\basis^{\subdivset^\star}_{\subpartition^\star}}%Basis depending on subdivisions and partition
\newcommand{\basisSBj}{\basis^{\subdivset}_{\blockset_j}}%Basis for one block depending on subdivisions and partition
\newcommand{\basisSi}{\basis_{s^{(i)}}}%basis created from subdivision s^(i)
\newcommand{\spaceESP}{E^\subdivset_{\subpartition}}%Space generated from the basis referred as basisSP
\newcommand{\setTSBj}{T^{\subdivset}_{\blockset_j}}%Space of knots generated from the subdivision referred as subdivset 
\newcommand{\projSP}{P^{\mathcal{S}}_{\mathcal{P}}}%Projection over the space referred as spaceESP

\newcommand{\hatfun}{\widehat{\phi}_{u,v,w}}%1-dim hat function
\newcommand{\hatfunSik}{\phi^{s^{(i)}}_k}%1-dim hat function defined by a sudvision s
\newcommand{\hatfunlj}{\phi_{\ul_j}}%multi-dim hat function defined by a block B_j and subdivisions 
\newcommand{\bPhiSBj}{\bPhi^{\subdivset}_{\blockset_j}}%Multi-dimensional function for one block
\newcommand{\bPhiSP}{\bPhi^{\subdivset}_{\subpartition}}%Multi-dimensional function for all blocks

\newcommand{\Lcritfull}{
\Lcrit((\subdivset,\subpartition),(\subdivset^\star,\subpartition^\star))}
\newcommand{\SEcrit}{\mathrm{SE}}

\newcommand{\MaxModCrit}{\mathcal{K}}
\newcommand{\Lcrit}{\mathrm{L2Mod}}
\newcommand{\SMSE}{\mathrm{SMSE}}

\newcommand{\convex}{\mathcal{C}}%Convex set
\newcommand{\domainD}{[0, 1]^D}%Domain of definition of the functions
\newcommand{\continuous}{\mathcal{C}^{0}}%Continuous functions 

\newcommand{\normal}{\mathcal{N}}%Normal law 
\newcommand{\trace}{\mathrm{Tr}}%trace operator
\newcommand{\var}{\widehat{\mathrm{VAR}}}%Var oper

\newcommand{\GP}{Y}%A GP
\newcommand{\kernl}{k}%The kernel associated to the GP
\newcommand{\BAGP}{\GP^{\subpartition}}%The block-additive GP
\newcommand{\kernlBAGP}{\kernl_{\subpartition}}%kernel associated with the BAGP 
\newcommand{\GPj}{\GP_j}%Block GP 
\newcommand{\kernlj}{\kernl_j}%kernel
\newcommand{\kernmat}{r}%Matern kernel
\newcommand{\GPSP}{\widetilde{\GP}^{\subdivset}_{\subpartition}}%finite-dimensional block-additive GP
\newcommand{\kernlSP}{\widetilde{k}^{\subdivset}_{\subpartition}}%Kernel associated to GPSP
\newcommand{\predictor}{\widehat{Y}^{\subdivset}_{\subpartition}}%Constrained predictor
\newcommand{\predictorstar}{\widehat{Y}^{\subdivset^\star}_{\subpartition^\star}}%New constrained predictor

\newcommand{\nknots}{|\setL^{\subdivset}_{\subpartition}|}

\newcommand{\Mstar}{\mathcal{M}^\star}

\title{Block-Additive Gaussian Processes under Monotonicity Constraints}

\author{
    Mathis Deronzier$^{1,3,\ast}$, Andr\'es F. L\'opez-Lopera$^{2}$, Fran\c{c}ois Bachoc$^{1,5}$, Olivier Roustant$^{3}$ \\
    and J\'er\'emy Rohmer$^{4}$
    \\\\
    \small ${}^{1}$Institut de Math\'ematiques de Toulouse (IMT), Univ. Paul Sabatier, F-31062 Toulouse, France.\\
    \small ${}^{2}$Univ. Polytechnique Hauts-de-France, CERAMATHS, F-59313 Valenciennes, France.\\ 
    \small ${}^{3}$IMT, UMR5219 CNRS, INSA, F-31077 Toulouse c\'edex 4, France.\\
    \small ${}^{4}$BRGM, 3 avenue Claude Guillemin, F-45060 Orléans c\'edex 2, France.\\
      \small ${}^{5}$Institut Universitaire de France (IUF).\\
    \small ${}^{\ast}$Corresponding author
}

\date{}

\begin{document}

\maketitle

\renewcommand{\contentsname}{Summary} 

%\tableofcontents 

\begin{abstract}
    We generalize the additive constrained Gaussian process framework to handle interactions between input variables while enforcing monotonicity constraints everywhere on the input space. The block-additive structure of the model is particularly suitable in the presence of interactions, while maintaining tractable computations. In addition, we develop a sequential algorithm, MaxMod, for model selection (i.e., the choice of the active input variables and of the blocks). We speed up our implementations through efficient matrix computations and thanks to explicit expressions of criteria involved in MaxMod. The performance and scalability of our methodology are showcased with several numerical examples in dimensions up to 120, as well as in a 5D real-world coastal flooding application, where interpretability is enhanced by the selection of the blocks.
\end{abstract}

\section{Introduction}\label{sec:intro}

%\noindent\textbf{Importance of (unconstrained) GP frameworks}
\paragraph{Constrained Gaussian processes (GPs).}
GPs are a central tool within the family of non-parametric Bayesian models, offering significant theoretical and computational advantages~\cite{Rasmussen2005GP}. They have been successfully applied in various research fields, including numerical code approximations \cite{Sacks89Design}, global optimization \cite{Jones1998EGO,bachoc2020gaussian}, model calibration \cite{kennedy2001bayesian},
geostatistics \cite{chiles2009geostatistics,mu2019intrinsic} and machine learning \cite{Rasmussen2005GP}. 

It is well-known that accounting for inequality constraints (e.g. boundedness, monotonicity, convexity) in GPs enhances prediction accuracy and yields more realistic uncertainties \cite{DaVeiga2012GPineqconst,DaVeiga2020GPineqconst,Pallavi2019BayesianShapeGPs,Wang2019DiffEq,Bachoc2019cMLE}. 
These constraints correspond to available information on functions over which GP priors are considered. Constraints such as positivity and monotonicity appear in diverse research fields, including social system analysis \cite{Riihimaki2010GPwithMonotonicity}, computer networking \cite{Golchi2015MonotoneEmulation}, econometrics \cite{Cousin2016KrigingFinancial}, geostatistics \cite{maatouk2017gaussian}, nuclear safety criticality assessment \cite{LopezLopera2017FiniteGPlinear}, tree distributions \cite{LopezLopera2019GPCox}, coastal flooding \cite{LopezLopera2019lineqGPNoise}, and nuclear physics \cite{Zhou2019ProtonConstrGPs}. The diversity of these domains highlights the versatility and relevance of constrained GPs. 

In this paper, we adapt the finite-dimensional framework of GPs introduced in~\cite{maatouk2017gaussian,LopezLopera2017FiniteGPlinear} to handle constraints using multi-dimensional ``hat basis'' functions locally supported around knots of a grid. 
In dimension one, the hat basis functions are also known as splines of degree one or $\mathbb{P}_1$ finite element basis functions. Importantly, this framework guarantees that constraints are satisfied everywhere in the input space.

However, even if recent improvements have been done to scale up this approach in the case of equally spaced knots \cite{maatouk2025large}, %when applying these constrained GP models to a target function $y : [0,1]^D \to \mathbb{R}$, 
one encounters the curse of dimensionality since the multi-dimensional hat basis functions are built by tensorization of the one-dimensional ones. \cite{Bachoc2022MaxMod} alleviates this issue by introducing the MaxMod algorithm, which performs variable selection and optimized knot allocation. MaxMod has been successfully applied to target functions up to dimension $D=20$ (though with fewer active variables).

\paragraph{Constrained additive GPs.} In the general statistics literature, a common approach to achieve dimensional scalability is to assume additive target functions:
\begin{equation}\label{eq:AdditiveFunction}
    y(x_1,\ldots, x_D) = y_1(x_{1}) + \cdots + y_{D}(x_{D}).
\end{equation}
Although this assumption may lead to overly “rigid” models, it results in simple frameworks that easily scale in high dimensions, as seen in~\cite{hastie2017generalized,buja1989linear} (without inequality constraints). Additive (unconstrained) GPs are considered in \cite{Durrande2012AdditiveGPs,Duvenaud2011AdditiveGPs}, as sums of one-dimensional independent GPs. We note that, besides computational advantages, the additive assumption also yields interpretability, such as the assessment of individual effects of input variables.

In \cite{lopez2022high}, constrained additive GPs are suggested based on the finite-dimensional approximation discussed above, providing a significant scaling to \cite{Bachoc2022MaxMod}, up to hundreds of dimensions. Furthermore, MaxMod has been adapted for variable selection and  knot allocation.

\paragraph{Extension to block-additive GPs (baGPs).}
In this paper, we seek a ``best of both worlds'' trade-off between \cite{Bachoc2022MaxMod}, which is more flexible but does not scale with dimension, and \cite{lopez2022high}, which scales better but cannot handle interactions between variables. 
Thus we suggest a block-additive structure, yielding block-additive GPs (baGPs). More precisely, we consider functions $[0,1]^D\to\R$:
\begin{equation}\label{eq:BlockAdditiveFunction}
    y(x_1,\ldots, x_D)=y_1(\bx_{\blockset_1})+\cdots+y_{B}(\bx_{\blockset_B}).
\end{equation}
Here, $\partition := \{\blockset_1, \ldots, \blockset_B\}$\label{page:subpartition}\label{page:blockset}\label{page:sizeblock} represents a subpartition of $\{1, \ldots, D\}$, where the disjoint union of the sets $\blockset_j$ is a subset of $\{1,\ldots,D\}$. The subset of variables $\bx_{\blockset_j}$ is simply obtained from $\bx$ by keeping the components of indices in $\blockset_j$.

The block-additive model offers flexibility in choosing the partition $\partition$, thereby encompassing both additive functions and functions with interactions at any order among all input variables. Its practical utility is especially relevant when the sizes of the blocks $|\blockset_j|$, equivalently the interaction orders, remain relatively small. In our constrained framework, this enhances the tractability of optimization and Monte Carlo sampling needed to compute the constrained GP posterior.

The construction of the finite-dimensional block-additive GP is not straightforward, as it requires to consider new bases and new methods to update them.
In practice, the block structure is unknown, but evaluations of the target function $y$ are available. A new challenge for this model is to infer the block-additive structure of $y$. Therefore, we propose a data-driven approach to select the blocks by providing an extension of the MaxMod algorithm. It is worth noting that outside of the GP world, the setting of block-additive models and methods for selecting blocks have been studied in the statistics literature, see for instance~\cite{storlie2011surface,stone1985additive,wood2017generalized}.
In particular, the ACOSSO method in \cite{storlie2011surface} is closest to GPs as it relies on reproducing kernel Hilbert spaces (RKHSs), but does not handle inequality constraints. 
Our extension of MaxMod is the first block selection method tailored to constrained GPs, to the best of our knowledge.

\paragraph{Summary of contributions.}  
In this paper, we consider a general target function $y$, known to belong to a convex set.
We focus on the convex set of componentwise monotonic (e.g. non-decreasing) functions. Nevertheless, as discussed in Remark \ref{remark:other_constr}, our framework can handle other convex sets for constraints such as componentwise convexity.

We make the following contributions.

\textbf{1)} We introduce a comprehensive framework for handling baGPs and constrained baGPs. Theoretical results are derived for multi-dimensional hat basis functions. In particular, we explicitly provide the change-of-basis matrices corresponding to adding active variables, merging blocks or adding knots.
We also use the matrix inversion lemma \cite[][Appendix A.2]{Rasmussen2005GP} to reduce the computational complexity.

\textbf{2)} We extend {MaxMod} to the block-additive setting, as discussed above. This algorithm maximizes a criterion based on the modification of the maximum a posteriori (MAP) predictor between consecutive iterations, hence its name MaxMod. For computational efficiency, we derive an explicit expression of the MaxMod criterion.

%Among the works focused on finite-dimensional approximation of GPs under constraints, only \cite{Bachoc2022MaxMod,lopez2022high} provide methods to identify the most active input variables and strategically place asymmetrical hat basis functions. Furthermore, \cite{lopez2022high} is the only work that accounts for inequality constraints in an additive GP framework. Similarly to those works, our framework also allows for asymmetrical hat basis functions and active variable selection. Importantly, our work is the first to address the selection of the partition in a constrained baGP model.

\textbf{3)} We provide predictors for every block-function $y_i$ in \eqref{eq:BlockAdditiveFunction} up to an additive constant (see Remark \ref{rem:blockdecomp}). The benefit for interpretability is highlighted on 
a real-world 5D coastal flooding problem previously studied \cite{azzimonti2019profile,LopezLopera2019lineqGPNoise,Bachoc2022MaxMod}.
%Specifically, we address the coastal flooding on the Boucheleurs district at La Rochelle. 
%Study of the constrained block-predictors $\widehat{y}_i$ discover the interactions between variables. 
%Indeed, this work makes the physical comprehension of coastal flooding more comprehensible. In general, the block-predictors improve the interpretability and makes the interpretation of physical behaviours easier.}

\textbf{4)} We demonstrate the scalability and performance of our methodology on numerical examples up to dimension 120. Our results confirm that MaxMod identifies the most influential input variables, making it efficient for dimension reduction while ensuring accurate models that satisfy the constraints everywhere on the input space. In the coastal flooding application, compared to \cite{Bachoc2022MaxMod}, our approach achieves higher accuracy with  fewer knots. 

\textbf{5)} We provide open-source codes that are integrated into the open-source R library \texttt{lineqGPR}~\cite{LineqGPR}.

\paragraph{Structure of the paper.} Section \ref{sec:Construction of the predictor} details the construction of the finite-dimensional baGPs. 
Section~\ref{sec:Conditioning a finite-dimensional baGP} explains how to handle the conditioning of a baGP to the inequality constraints and the observations.
Section \ref{sec:MaxMod algorithm} introduces the MaxMod algorithm. Section \ref{sec:Numerical results} presents the numerical results on toy examples and the 5D coastal flooding application. Finally, Section \ref{sec:conclusions} summarizes the conclusions and potential future work. The proofs and additional content are provided in the Appendix.

\section{baGPs and their finite-dimensional approximations} \label{sec:Construction of the predictor}

In this section we consider a fixed subpartition $\subpartition=\{\blockset_1,\ldots,\blockset_B\}$ of $\{1,\ldots,D\}$ and we delve into the construction of the finite-dimensional baGP predictor.
This construction relies on two steps. Firstly, we introduce an infinite-dimensional baGP that is refereed to as $\BAGP$. Secondly, for each block we construct a family of hat basis functions. The finite-dimensional GP is then obtained by projection of $\BAGP$ onto the vector space spanned by them. Table \ref{table:list:of:symbols} provides the list of the main notation symbols for Sections \ref{sec:Construction of the predictor} and \ref{sec:Conditioning a finite-dimensional baGP}.  

\subsection{Block-additive GPs}
\label{subsec:Unconstrained baGPs}

For each $1\leq j\leq \sizeblock$, we consider a centered GP $\{\GPj(\bx), \, \bx \in [0,1]^{|\blockset_j|}\}$\label{page:GPj} with kernel $\kernlj$\label{page:kernlj}. 
% An example of a kernel $k_j$\label{page:kernlj} is
% \begin{equation}\label{eq:kernel}
% 	\kernlj (\bx_{\blockset_j},\bx_{\blockset_j}')=\sigma^2_j \prod_{i\in \blockset_j} \kernmat_{\theta_{i}}(x_i,x'_i),
% \end{equation}
% where $\bx_{\blockset_j}=(x_i)_{i \in \blockset_j}$, $\bx'_{\blockset_j}=(x'_i)_{i \in \blockset_j}$, and $\kernmat_{\theta_{i}}$ are the normalized one-dimensional Mat\'ern kernel:
% \[
%     \kernmat_\theta(x,x')= \left(1+\sqrt{5}\frac{|x-x'|}{\theta}+ \frac{5}{3} \frac{|x-x'|^2}{\theta^2}\right)\exp\left(-\sqrt{5} \frac{|x-x'|}{\theta}\right).
% \]
% For this particular kernel structure, each block has one variance parameter $\sigma_j^2 \in \mathbb{R}^+$ and one length-scale parameter per dimension, denoted $\theta_{i} \in \mathbb{R}^+$. This results in a total of $B+D$ covariance parameters. 
We then define the baGP $\BAGP$ as 
\begin{equation}\label{eq:BAGP}
\BAGP(\bx)= Y_1(\bx_{\blockset_1})+ \cdots + Y_{\sizeblock}(\bx_{\blockset_\sizeblock}).
\end{equation}
Assuming that
$(\GPj)_{1\leq j\leq \sizeblock}$ are independent, then
$\BAGP$ \label{page:BAGP} is also a centered GP with kernel $\kernlBAGP:\domainD\times \domainD\to \R$\label{page:kernlBAGP} satisfying
\begin{equation}\label{eq:kernlBAGP}
    \kernlBAGP(\bx,\bx')=\sum_{j=1}^{B} \kernlj (\bx_{\blockset_j},\bx_{\blockset_j}').
\end{equation}
An example of kernel $k_j$ is
\begin{equation}\label{eq:kernel}
	k_j(\bx_{\blockset_j},\bx_{\blockset_j}')=\sigma^2_j \prod_{i\in \blockset_j} \kernmat_{\theta_{i}}(x_i,x'_i),
\end{equation}
where $\bx_{\blockset_j}=(x_i)_{i \in \blockset_j}$, $\bx'_{\blockset_j}=(x'_i)_{i \in \blockset_j}$, and for all $\theta \in \mathbb{R}^+$,  $\kernmat_{\theta}$  is the one-dimensional Mat\'ern correlation kernel:
   \vspace{-0cm}
\[
    \kernmat_\theta(x,x')= \left(1+\sqrt{5}\frac{|x-x'|}{\theta}+ \frac{5}{3} \frac{|x-x'|^2}{\theta^2}\right)\exp\left(-\sqrt{5} \frac{|x-x'|}{\theta}\right).
       \vspace{-0cm}
\]
For this particular kernel structure, each block has one variance parameter $\sigma_j^2 \in \mathbb{R}^+$ and one length-scale parameter per dimension, denoted $\theta_{i} \in \mathbb{R}^+$. This involves at most of $B+D$ covariance parameters.

Each $\GPj$ is a Gaussian prior over the function $y_j$ defined in \eqref{eq:BlockAdditiveFunction}, then $\BAGP$ is the Gaussian prior over the latent function $y = y_1 \oplus \cdots \oplus  y_{B}$. Note that handling the functional constraint $\BAGP \in \convex$ is the strongest challenge of constrained GPs. %\cite{DaVeiga2012GPineqconst} managed to provide approximations and methods for computing moments of this law in specific cases of convex sets. 
To make this possible, we approximate  $\BAGP$ by a finite-dimensional GP, enabling to characterize the (functional) constraints by equivalent finite-dimensional ones. 

\subsection{Hat basis functions and monotonicity constraints}\label{subsec:need:of:hat:basis}

In Section \ref{subsec:finite:dim:approx},
we approximate a GP by a finite-dimensional one living in the vector space $E$ spanned by hat basis functions. The use of these functions has been developed in several articles \cite{bay2017new,Bachoc2022MaxMod,lopez2022high}.
Figure \ref{fig:hat_basis} shows an example of a one-dimensional hat basis $\{\phi_1, \ldots, \phi_5\}$ and the projection of a monotonic function on its corresponding vector space $E$.
In $E$ we have an equivalence between monotonicity of a function and its values at the knots. Basically, a piecewise affine function is non-decreasing if and only if the sequence of values at the knots is non-decreasing. 
\begin{figure}[t!]
\centering
\begin{tabular}{cc}
\includegraphics[width=0.48\linewidth]{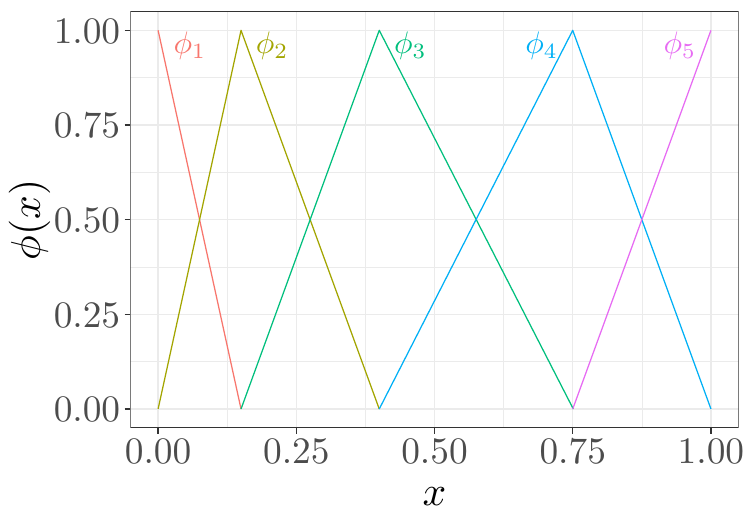} &
\includegraphics[width=0.48\linewidth]{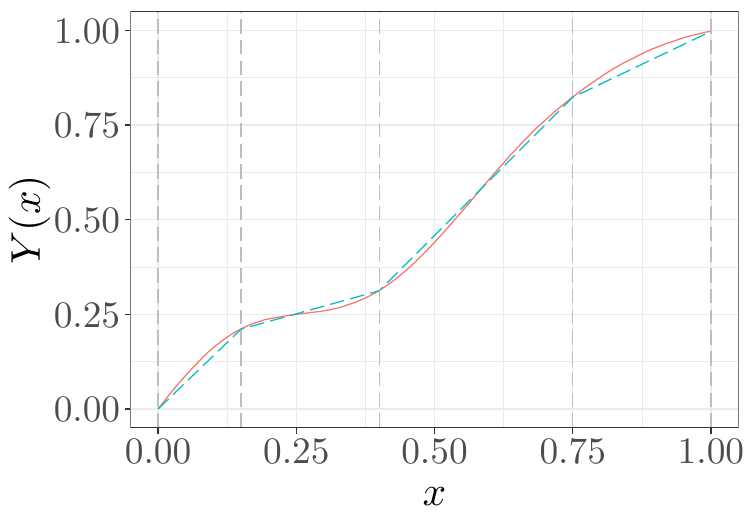} 
\end{tabular}
\caption{The panels show an example of (left) a one-dimensional hat basis generated from the subdivision $s=(0, 0.1, 0.2, 0.5, 0.85, 1)$, and (right) an example of the projection (in blue) of a monotonic function (in red) onto the corresponding vector space.}
\label{fig:hat_basis}
\end{figure}
In other words, for any function $y = \sum_{i=1}^m a_i \phi_i$, we have the following equivalence:
\begin{equation}\label{eq:monot:equiv}
	y \mbox{ is monotonic} \iff a_i\leq a_{i+1},\, \forall i\in \{1,\ldots, m-1\}.
\end{equation}
We then transformed a functional constraint into a linear constraint in a finite-dimensional space.

\subsection{Finite-dimensional approximation} \label{subsec:finite:dim:approx}
We first define the hat basis functions, starting from the one-dimensional case. Then, we define the corresponding finite-dimensional approximation of $\BAGP$, obtained by projection.% Figure \ref{fig:hat_basis} provides an example of hat basis family and the projection of a function over the space spanned by these functions.

\subsubsection{One-dimensional hat basis functions}
\label{subsubsec:1DimensionalBasis}

The one-dimensional hat basis functions
are defined from a subdivision of $[0,1]$. Let $s$ be this subdivision, $s=(t_1,\ldots,t_m)$ \label{page:subdivision} with $t_1=0< \dots <t_m=1$. We call $t_1, \ldots, t_m$ one-dimensional knots and $m$ the size of the subdivision. We write $\widehat{\phi}_{u,v,w}:[0,1]\to \R$\label{page:hatfun}, for $ u < v < w$, the hat function with support $[u,w]$ having two linear components on $[u,v]$ and $[v,w]$, and equal to $1$ at $v$. The function is defined as 
\begin{equation*}
\label{eq:hatfun}
\widehat{\phi}_{u,v,w}(x)=
\begin{cases}
\dfrac{x-u}{v-u} & \mbox{if } u\leq x\leq v,\\
\dfrac{w-x}{w-v} & \mbox{if } v\leq x \leq w,\\
0 & \mbox{otherwise}.
\end{cases}
\end{equation*}
The basis created by the subdivision $s$ is then 

\begin{equation}\label{eq:1-dimHatBasis}
\basis_s=\{\phi^s_1, \ldots, \phi^s_m\}, \hspace{1cm} \phi^s_i:=\widehat{\phi}_{t_{i-1},t_{i},t_{i+1}},\, 1\leq i\leq m,
\end{equation} 
setting $t_{0}=-1$ and $t_{m+1}=2$ by convention. Note that for every $u<v<w$, $\widehat{\phi}_{u,v,w}$ can be seen as a function from $[0,1]$ just by considering its restriction to the segment.

\subsubsection{Multi-dimensional hat basis functions}\label{subsubsec:MultiDimensionalBasis}

Denote the set $X=X^{(1)}\times \cdots \times X^{(D)}$ with $X^{(i)}=[0,1]$. For $i=1,\ldots,D$,  let $s^{(i)}=(t^{(i)}_1,\ldots, t^{(i)}_{m^{(i)}})$ be a subdivision of $[0,1]$ and define the set of all subdivisions, $\subdivset=(s^{(1)},\ldots, s^{(D)})$\label{page:subdivisions}. From \eqref{eq:1-dimHatBasis}, each subdivision $s^{(i)}$ generates a hat basis $\basisSi:=\{\hatfunSik,\, k= 1, \ldots, m^{(i)}\}$\label{page:basisSi}\label{page:hatfunSik}. For each $j \in \{1,\ldots, \sizeblock\}$, consider the block $\blockset_j$. We can define multi-dimensional functions obtained by tensorizing one-dimensional hat bases $\basisSi$. We introduce the set of multi-indices \label{page:multiset_j}
\begin{equation}\label{eq:multiset_j}
	\setLSBj =\prod_{i\in \blockset_j} \left\{1,\ldots,m^{(i)}\right\}=\left\{\ul_j=(\ul_{j,i})_{i \in \blockset_j}, \, 1 \leq \ul_{j,i} \leq m^{(i)}, \, \forall i \in \blockset_j \right\}.
\end{equation}
For every element $\ul_j$\label{page:ulj} in $\setLSBj$ corresponds a multidimensional hat-function $\phi_{\ul_j}:X\to \R$,\label{page:hatfunlj}
\begin{equation}
\label{eq:phi:SBj}
\phi_{\ul_j}(\bx)=\prod_{i \in \blockset_j} \phi^{s^{(i)}}_{\ul_{j,i}}(x_i).
\end{equation}
Note that for a fixed $j$, the functions $\phi_{\ul_j}$ essentially depend on the variables $(x_i)_{i\in \blockset_j}$. We denote by $\continuous(X^{\blockset_j}, \R)$\label{page:continuous:blocksetj} the set of continuous functions depending only on variables indexed by $\blockset_j$. Then, the following inclusions hold
\begin{equation}\label{eq:inclusion:function:space}
   \{\phi_{\ul_j},\, \ul_j \in \setLSBj\}\subset \continuous(X^{\blockset_j}, \R) \subset \continuous(X, \R).
\end{equation}
For a hat function $\phi_{\ul_j}$, we consider the point such that $\phi_{\ul_j}(\bt_{\ul_j})=1$ corresponding to the top of the hat,
\begin{equation}
\label{eq:t}
\bt_{\ul_j}=\left(t^{(i)}_{\ul_{j,i}}\right)_{i \in \blockset_j}.
\end{equation}
Finally, we define the vector space $\spaceESP$\label{page:spaceESP} and the multiset index $\setL^{\subdivset}_{\subpartition}$ as
\begin{equation}\label{eq:spaceESP:setLSP}
	\spaceESP = 
	\operatorname{span}\left(\phi_{\ul_j}\right)_{\ul_j \in \setLSBj, 1\leq j \leq\sizeblock},
	\hspace{1cm} 
	\setLSP =
	\bigcup_{j=1}^{B} \setLSBj.
\end{equation}

\subsection{Projection and finite-dimensional GPs}
\label{subsec:Projection Finite Dimension GP}

Given a subpartition $\subpartition=\{\blockset_1, \cdots, \blockset_\sizeblock\}$ and subdivisions $\subdivset$, a projection $\projSP$\label{page:projSP} over the space $\spaceESP$ can be defined as
\begin{equation}\label{eq:ProjectionFunction}
\begin{array}{cccc}
\projSP:& \continuous(X^{\blockset_1},\R) + \cdots +\continuous(X^{\blockset_\sizeblock},\R) &\to & \spaceESP \\
& \sum_{j=1}^{\sizeblock} f_j & \mapsto & \sum_{j=1}^{\sizeblock}\sum_{\ul_j \in \setL^{\subdivset}_{\blockset_j}} f_j(t_{\ul_j})\phi_{\ul_j}
\end{array}.
\end{equation}
In the above equation the sets $\continuous(X^{\blockset_j}, \R)$ are the ones defined in~\eqref{eq:inclusion:function:space}. Recall that $\BAGP$ is the block-additive GP defined in \eqref{eq:BAGP}. We define the centered finite-dimensional baGP $\GPSP$\label{page:GPSP} as
\begin{equation}\label{eq:GPSP}
\GPSP(\bx)=\projSP(\BAGP)(\bx) = \sum_{j=1}^{\sizeblock}\sum_{\ul_j \in \setL^{\subdivset}_{\blockset_j}} \GP_j(t_{\ul_j})\phi_{\ul_j}(\bx).  
\end{equation}
Its kernel $\kernlSP$\label{page:kernlSP} is then given by
\begin{equation}\label{eq:kernlGPSP}
    \kernlSP(\bx,\bx')=\sum_{j=1}^{B}\sum_{\ul_j, \ul_j' \in\setL^{\subdivset}_{\blockset_j}} \kernlj(t_{\ul_j},t_{\ul'_j}) \phi_{\ul_j}(\bx)\phi_{\ul'_j}(\bx').
\end{equation}

\begin{remark}\label{rem:blockdecomp}
Notice that the application $\projSP$ in \eqref{eq:ProjectionFunction} is well defined, meaning that $\projSP(f)$ is unique although $f$ can be written in several manners $f= \sum_{j=1}^B f_j$. To see this, assume that $f= \sum_{j=1}^B f_j = \sum_{j=1}^B g_j$ where $f_j, g_j \in \continuous(X^{\blockset_j}, \R)$. Recall that the blocks $\blockset_1, \dots, \blockset_B$ are disjoint. As $f_j, g_j$ depend only on the variables in $\blockset_j$, setting to $0$ all the variables that are not in $\blockset_j$, we can see that $g_j - f_j$ is equal to some constant $u_j = \sum_{j' \neq j} (f_{j'}(0) - g_{j'}(0)) $. Furthermore, as $\sum_{j=1}^B (g_j - f_j) = 0$ we must have $\sum_{j=1}^B u_j = 0$. Now, 
%Indeed, the decomposition of a function in $\continuous(X^{\blockset_1},\R) + \cdots +\continuous(X^{\blockset_\sizeblock},\R)$ is not unique as 
%\begin{equation}\label{eq:sum:decomposition}
%f= \sum_{j=1}^\sizeblock f_j+u_j, \, \text{whenever} \;  \sum_{j=1}^{\sizeblock}u_j=0. 
%\end{equation}
%However, 
as for any $j$, $\sum_{\ul_j\in \setLSBj} \phi_{\ul_j}=1$, 
%the following equality holds
\begin{align*}
    \projSP\left(\sum_{j=1}^\sizeblock g_j\right) =
 %          \projSP\left(\sum_{j=1}^\sizeblock f_j+u_j\right)=
 \sum_{j=1}^{\sizeblock}\sum_{\ul_j \in \setLSBj} (f_j(t_{\ul_j})+u_j)\phi_{\ul_j} 
    =\sum_{j=1}^{\sizeblock}\sum_{\ul_j \in \setLSBj} f_j(t_{\ul_j})\phi_{\ul_j}+\sum_{j=1}^{\sizeblock} u_j
%= \projSP(f)    
    = \projSP\left(\sum_{j=1}^\sizeblock f_j\right),
\end{align*}
which shows that $\projSP(f)$ is uniquely defined.
\end{remark}

\section{Conditioning a finite-dimensional baGP}
\label{sec:Conditioning a finite-dimensional baGP}

In this section, we assume that are given a subpartition $\subpartition=\{\blockset_1, \ldots, \blockset_\sizeblock\}$, subdivisions $\subdivset=(s^{(1)},\ldots, s^{(D)})$ and the associated finite-dimensional baGP $\GPSP$.

This section is dedicated to find the law of the baGP $\GPSP$ constrained to the (possibly noisy) observations $(\GPSP(\bx_i)+ \epsilon_i= y_i)_{i=1,\ldots,n}$ and the functional constraint $\GPSP\in\convex$. Defining $\bX=[\bx_1, \ldots, \bx_n]^\top$, $\bY=[y_1, \ldots, y_n]^\top$ and $\bepsilon= [\epsilon_1, \ldots,\epsilon_n]^\top$ a Gaussian noise of law $\normal(0, \tau^2\bI_n)$ independent with $\GPSP$, then we aim to study
\[
\left(\GPSP(\bx)\,\big|\,
\GPSP(\bX)+ \bepsilon = \bY,\,
\GPSP \in \convex \right).  
\]
We use above classical notations in GP framework, for which we give a reminder here. Given two sets $A$ and $B$, for any vector $\ba=[a_1,\ldots,\, a_n]^\top \in A^n$, $\bb=[b_1,\dots, b_m]^\top \in B^m$ and any function $f: A\times B \to \R$, the notation  $f(\ba,\bb)$ corresponds to the matrix in $\R^{n\times m}$, $f (\ba,\bb)=(f (a_i,b_j))_{1\leq i\leq n, \\ 1\leq j \leq m}$.

We first show that it is equivalent to work on conditioning a Gaussian vector $\bxi$. 
Then, we condition this vector to the interpolations constraints. Finally, we show the equivalence between the functional constraint of our finite-dimensional baGP, $\GPSP\in \convex$, and a finite-dimensional spatial constraint of our Gaussian vector $\bxi\in\convex'$. This unable to condition by inequality constraints.

We found a more efficient way to compute the law of the conditioned Gaussian vector $\bxi$ to the observations. As the algorithm complexity of our method relies on this computation, we will discuss the complexity improvements of our method.

\subsection{Boiling down to conditioning a Gaussian vector}\label{subsec:conditioning on Gaussian vector}

%This section provides details about the computation of the predictor constructed by finite-dimensional functions. We shall determine the distribution of the random vector $\bxi$ conditioned on the observations $(\bX,\bY)$. We provide explicit ways to compute the conditional mean and covariance matrix using classical tricks such as the Woodbury formula \cite{Rasmussen2005GP} and the Cholesky decomposition of matrices.

%\subsubsection{}\label{subsubsec:Computation of the covariance matrix}

Given an order on the elements of $\setLSBj$ for $j=1, \ldots, \sizeblock$, we define the multi-dimensional function $\bPhi:=\bPhiSP$ as

\begin{equation}
\label{eq:bPhiSP}
\begin{array}{cccc}
\bPhi: & \domainD & \longrightarrow & \R^{|\setL^{\subdivset}_{\blockset_1}|+ \cdots + |\setL^{\subdivset}_{\blockset_B}|}\\
&\bx & \mapsto & [\bPhi^{\subdivset}_{\blockset_1}(\bx)^\top, \ldots, \bPhi^{\subdivset}_{\blockset_\sizeblock}(\bx)^\top]^\top
\end{array},
\end{equation}
where $\bPhiSBj = (\phi_{\ul_j})_{\ul_j \in \setLSBj}$ is a column vector function. Recall the (infinite-dimensional) baGP defined in \eqref{eq:BAGP} is $Y^\subpartition= \sum_{j=1}^{\sizeblock} \GPj$, with $\GPj \sim \mathrm{GP}(0,\kernlj)$.
From \eqref{eq:GPSP} we can rewrite $\GPSP$ as the scalar product of a Gaussian vector $\bxi$ and the multidimensional function $\bPhi$
\begin{equation}
\label{eq:GPSP:scalar:prod}
\GPSP= \bPhi^\top\bxi, \hspace{1cm} \bxi =(\bxi_1, \ldots, \bxi_B)^\top,
\end{equation} 
\begin{equation}\label{eq:Xij}
\bxi_j=\GPj(\bt_j)
\sim \normal\big(0, \kernlj (\bt_j,\bt_j)\big), 
\hspace{1cm} \bt_j=(t_{\ul_j})_{\ul_j \in \setL^{\subdivset}_{\blockset_j}}.
\end{equation} 
\label{page:Xi}\label{page:Xij}The independence hypothesis between the GPs $\GPj$ implies independence between the vectors $\bxi_j$, hence the covariance matrix $\widetilde{\bK}$ of $\bxi$ is a block-diagonal matrix with blocks $(\kernlj (\bt_j,\bt_j))_{j=1}^{\sizeblock}$. Moreover, the zero-mean hypothesis on $\GP_j$ implies $\bxi \sim \normal(0, \widetilde{\bK})$. Linearity of the conditioning allows us to write
\[
\left(\GPSP(x)\big|\GPSP(\bX)+\bepsilon=\bY, \GPSP\in \convex \right) =\bPhi(\bx)\left(\bxi\big|\GPSP(\bX)+\bepsilon=\bY, \GPSP\in \convex \right),
\]
underlying that we only need to work on the conditioning of the Gaussian vector $\bxi$.

\subsection{Interpolation constraints and computation costs}\label{subsec:Interpolation constraints and computation cost}
The conditional Gaussian vector $(\bxi| \,\bPhi(\bX)^\top \bxi+\bepsilon=\bY)$, where $\bepsilon \sim \normal(0, \tau^2\bI_n)$ is independent of $\bxi$, has mean $\bmu$ with
\begin{equation}
\bmu = \widetilde{\bK}\bPhi(\bX) \Big[\bPhi(\bX)^\top \widetilde{\bK}\bPhi(\bX) + \tau^2 \bI_n\Big]^{-1}\bY.
\label{eq:muCondXi}
\end{equation}
The computation of $\bmu$ has been studied in \cite[][Appendix 2]{lopez2022high} when $n \ll \nknots$ using the Woodbury identity, also known as the matrix inversion lemma (see \cite[][Appendix 3]{Rasmussen2005GP}). A significant speed up is obtained to compute $\big[\bPhi(\bX)^\top \widetilde{\bK}\bPhi(\bX) + \tau^2 \bI_n\big]^{-1}$ in $\complexity(\nknots^3+n\nknots^2)$, compared to $\complexity(n^3+n^2|\setL^{\subdivset}_{\subpartition}|+n|\setL^{\subdivset}_{\subpartition}|^2)$ which is the complexity of the direct computation. The covariance $\condK$ of the conditional Gaussian vector $(\bxi| \,\bPhi(\bX)^\top \bxi+\bepsilon=\bY)$ is
\begin{equation}
\condK = \widetilde{\bK}-\widetilde{\bK}\bPhi(\bX) \Big[\bPhi(\bX)^\top \widetilde{\bK}\bPhi(\bX)+ \tau^2 \bI_n \Big]^{-1} \bPhi(\bX)^\top \widetilde{\bK}.
\label{eq:SigmaCondXi}
\end{equation}
The direct computation of $\condK^{-1}$, required in the MAP estimation detailed in Section~\ref{subsubsec:Convex handled by the algorithm}, has a complexity of $\complexity (\nknots^3+ n^3)$ due to the two matrices inversions. In this paper we use an alternative formula for the computation of $\condK^{-1}$ provided again by the Woodbury identity:
\begin{equation}\label{eq:Woodbury inversion}
    \condK^{-1} = \widetilde{\bK}^{-1}+\tau^{-2}\bPhi(\bX)\bPhi(\bX)^\top.
\end{equation}
The block diagonal structure of $\widetilde{\bK}$ allows the computation of $\condK^{-1}$ in $\complexity\big( \sum_{j=1}^{\sizeblock}|\setL^{\subdivset}_{\blockset_j}|^3+\nknots^2n\big)$. This complexity stems from the inversion of each block of $\widetilde{\bK}$ followed by the computation of $\bPhi(\bX)\bPhi(\bX)^\top$. This improvement in the complexity underlines the fact that block-additive structures allow us to deal with ``independent problems'' in smaller dimension. Therefore, to compute $\condK^{-1}$, it is preferable to use \eqref{eq:Woodbury inversion} instead of \eqref{eq:SigmaCondXi}. Table~\ref{tab:complexity} summarizes some of the computational costs involved in the computation of the conditional finite-dimensional GP when $n \gg \nknots$ and $n \ll \nknots$.

\begin{table}[t!]
    %\scriptsize
    \centering
    \caption{Illustration of the complexity cost of computation in different cases. Notations $*_D$ and $*_W$ hold respectively for \textbf{Direct} or the \textbf{Woodbury} computation method.}
    \label{tab:complexity} 
    \setlength{\tabcolsep}{3.8pt}
    {\small
    \begin{tabular}{c|cccccc}
        \toprule
        & \multicolumn{6}{c}{Complexity computation} \\
        & $\bmu_{D}$ & $\bmu_{W}$ & $\condK^{-1}_{D}$ & $\condK^{-1}_{W}$ & $(\bmu_{D}, \condK^{-1}_{W})$ & $(\bmu_{W}, \condK^{-1}_{W})$\\
        \midrule
        $n \gg \nknots$ & 
        $\complexity(n^3)$ &
        $\complexity(n^2 \nknots)$ &
        $\complexity(n^3)$ &
        $\complexity(n^2\nknots)$ & $\complexity(n^3)$ &
        $\complexity(n^2\nknots)$  \\
        
        $n \ll \nknots$ &
        $\complexity(n\nknots^2)$  & 
        $\complexity(\nknots^3)$ &  $\complexity(\nknots^3)$ & $\complexity(n\nknots^2+\sum |\setLSBj|^3 )$ & 
        $\complexity(n\nknots^2+\sum |\setLSBj|^3 )$ &
        $\complexity(\nknots^3)$ \\
        \bottomrule
    \end{tabular}
    }
\end{table}

%Discussion over the complexity improvement
%To illustrate improvement of the complexity obtained by the computation of $\condK^{-1}$ when $n \ll \nknots$, suppose that there is the same number of functions in families $\basisSBj$, so that  $|\setL^{\subdivset}_{\blockset_j}| = \nknots/\sizeblock$ for $j = 1, \ldots, \sizeblock$. Suppose as well that the size of the block is bounded by $d\ll D$, meaning that we are in the case where our model is fitted for the studied function, thus $B\geq D/d$. Then, $\condK^{-1}$ can be computed in $\complexity(d^2\nknots^3/D^2) $ instead of $O(\nknots^3)$. This last complexity shows that richness of the model can increase almost linearly with the dimension of the problem without generating computational issues.

\begin{remark}
In \cite{lopez2022high}, authors deal with the additive model corresponding of blocks of size $1$. As they use the direct computation of $\condK^{-1}$, it leads to a complexity of $\nknots^3$ when $n\ll \nknots$. From what we show in table~\ref{tab:complexity} and what we said, using \eqref{eq:Woodbury inversion} would lead to a significant improvement on the complexity of the model. If the size of the bases are all equal: $|\setLSBj|=|\setLSP|/\sizeblock$. The complexity with the Woodbury formula in \eqref{eq:Woodbury inversion} is $\complexity(|\setLSBj|^3/\sizeblock^2)$ instead of $\complexity(|\setLSBj|^3)$, which is particularly interesting in high dimension.
\end{remark}

%\subsubsection{Instability issues}

%		\textcolor{red}{La decomposition de Cholesky s'ecrit $L L^\top$ avec L triangulaire inférieure ou bien $U^\top U$ avec U triangulaire superieure. A adapter dans le document.}

%\textbf{Instability for the inversion when } $\nknots 
%\ll n$:\\
%In this section we will focus on the instabilities that can be caused by inversion and conditioning of matrices. We need to compute
%\[
%\bC^{-1}= \Big[\bPhi(\bX)^\top \widetilde{\bK}\bPhi(\bX) + \tau^2 \bI_n\Big]^{-1},
%\]
%writing the Cholesky decomposition $\widetilde{\bK}=\bL^\top\bL$ , we can rewrite using the Woodbury formula :
%\begin{align*}
%\bC^{-1}
%&= \Big[ \big(\bL\bPhi(\bX)\big)^\top \bI_{\nknots}\big(\bL\bPhi(\bX)\big) + \tau^2 \bI_n\Big]^{-1}
%\\
%&= \tau^{-2}\Big[ \bI_n - \big(\bL\bPhi(\bX)\big)^\top\Big(\tau^2\bI_{\nknots} + \big(\bL\bPhi(\bX)\big)\big(\bL\bPhi(\bX)\big)^\top\Big)^{-1} \big(\bL\bPhi(\bX)\big) \Big].
%\end{align*}
%Let the Cholesky decompostition of $\tau^2\bI_{\nknots} + \big(\bL\bPhi(\bX)\big)\big(\bL\bPhi(\bX)\big)^\top = \widetilde{\bL} \widetilde{\bL}^\top$. 
%We can rewrite the last equality 
%\[
%\bC^{-1}= \tau^{-2}\Big[ \bI_n - \big(\widetilde{\bL}^{-1}\bL\bPhi(\bX)\big)^\top\big(\widetilde{\bL}^{-1}\bL\bPhi(\bX)\big)\Big],
%\]
%Since  in the last equality $\widetilde{\bL}$ is triangular, the system $\widetilde{\bL}\bM = \bL\bPhi(\bX)$ can be solved for $\bM$ in $\complexity(\nknots^2)$.

\subsection{Verifying inequality constraints everywhere with finite-dimensional baGPs} \label{subsubsec:Convex handled by the algorithm}

Recall that the set $\convex$ of (componentwise) monotonic functions is a subset of $\continuous([0,1]^D,\R)$. Note that, even if a constrained GP model has a subset $J$ of active variables that is strictly smaller than $D$, it can be considered as a process of the full $D$ variables, and considered as such, it is required to belong to $\convex$.

For a given subpartition $\subpartition$ and subdivisions $\subdivset$, the method to construct the predictor $\predictor$ is to take the mode of the finite-dimensional GP $\GPSP = \bPhi^\top \bxi$, conditioned by the observations $\GPSP(\bX) + \bepsilon= \bY $ and the condition $\GPSP \in \convex$. Here $\bPhi$ is the multi-dimensional function defined in~\eqref{eq:bPhiSP}. The hat basis family, contrary to other spline families, allows to find a convex subset $\convex'\subset \R^{|\setLSP|}$ such that the following equivalence holds 
\begin{equation}
\label{eq:EquivalenceCCp}
\GPSP \in \convex \Longleftrightarrow \bxi \in \convex'. 
\end{equation}
Details of the characterization of $\convex'$ are given below.
Hence, finding the mode of the truncated Gaussian vector $(\bxi|\, \GPSP(\bX) + \bepsilon = \bY, \, \GPSP \in \convex)$, is equivalent to solve the minimization problem
\begin{equation}
\label{eq:xi:mod:sol}
    \widehat{\bxi}=\underset{\bxi \in \convex'}{\arg\min}\, (\bxi-\condmu)^\top \condK^{-1}(\bxi-\condmu).
\end{equation}
In~\eqref{eq:xi:mod:sol}, $\condmu$ and $\condK$ are the mean and the covariance matrix of the Gaussian vector $(\bxi | \, \GPSP(\bX) + \bepsilon = \bY)$ detailed in \eqref{eq:muCondXi} and \eqref{eq:SigmaCondXi}. The mode predictor\label{page:predictor} is then
\begin{equation}\label{eq:predictor}
\predictor = \bPhi^\top \widehat{\bxi}.
\end{equation}
From \eqref{eq:xi:mod:sol}, we see that the difficulty of finding the solution depends on the difficulty of handling a quadratic minimization problem on the set $\convex'$. 
%\subsubsection{Explicit link between convex sets and linear inequality constraints}
%\label{subsection:equivalence:constraints}

Now, let us consider the  set $\convex$ of monotonic functions, as in the rest of the paper. Then, $\convex'$ can be made explicit.
%I believe we need to better motivate this choice... Sounds arbitrary as it is presented... We can maybe add a discussion about this choice and refer to Remark \ref{remark:other_constr})}  
%Then a finite-dimensional block-additive GP based on hat basis functions is in  $\convex$ if and only if the vector of coefficients is in a known polyhedral set $\convex'$.
%
Let us first explain the case $D=1$, which is simplest to expose. 
Let $s = (t_1, \ldots, t_m)$ and $\basis_s =\{\phi^s_1, \ldots, \phi^s_m\}$ be the subdivision and its associated hat basis defined in~\eqref{eq:1-dimHatBasis}. For any function $f$ written as a linear combination of elements in $\basis_s$, $f = \sum_{i=1}^m a_i \phi^s_i$, we have the following equivalence developed in Section \ref{subsec:need:of:hat:basis}:
$$f \mbox{ is monotonic if and only if, for any } 1\leq i\leq m-1, \ a_i\leq a_{i+1}.$$
These inequalities can be rewritten as linear inequalities. Letting $\ba=[a_1,\ldots, a_m]^\top$, then there is a matrix $\bLambda \in M_{m-1,m}$ such that $f$ is monotonic if and only if $\bLambda \ba \leq 0$. 
The case of a general value of $D$ shares some ideas with the case $D=1$, but the explicit linear inequalities are more cumbersome to express.
Note that since we consider non-overlapping blocks, a block-additive function is monotonic if and only if all the individual block functions are monotonic. Then for a function of the form 
\[
\sum_{j=1}^{\sizeblock}\sum_{\ul_j \in \setL^{\subdivset}_{\blockset_j}} a^j_{\ul_j} \phi_{\ul_j}
\]
as in \eqref{eq:ProjectionFunction}, all the functions $\sum_{\ul_j \in \setL^{\subdivset}_{\blockset_j}} \alpha^j_{\ul_j} \phi_{\ul_j}$ must be monotonic. Hence, the set of linear inequalities defining $\convex'$ is of the form $\bLambda \ba \leq 0$, where $\bLambda$ is block diagonal composed of the $\sizeblock$ blocks  $\bLambda_1,\ldots,\bLambda_\sizeblock$ and $\ba$ concatenating the $a^j_{\ul_j}$'s is written as $[\ba_1^\top,\ldots,\ba_\sizeblock^\top]^\top$, with the same dimensions. The expressions of the $\bLambda_i$'s are given in the supplementary material of \cite{Bachoc2022MaxMod} (Section SM1). 

%The new challenge for block-additive functions defined in Equation \eqref{eq:BlockAdditiveFunction} is the boundedness constraint. This constraint provides nonlinear inequalities. To deal with these types of constraints, we have introduced a new construction of the convex set describing bounded additive GPs in the finite-dimensional representation (see Appendix \ref{appendix:Boundedness Constraint}). More precisely, we redefined the convex set in terms of nonlinear inequality constraints and detailed how to get a projection over this set seeking to apply optimization algorithms to this set.

\begin{remark}
    \label{remark:other_constr}
    In this paper, we only focus on monotonic functions but in all generality the optimization problem we are able to solve, as in \eqref{eq:xi:mod:sol}, are ones on polyhedra, that are the sets defined by $\ba\in \convex'$ if and only if $\bLambda \ba \leq \bx$, a topic further explored in \cite{LopezLopera2019lineqGPNoise, maatouk2017gaussian}.
    Similar equivalences as in \eqref{eq:EquivalenceCCp} can be obtained when $\convex$ is the set of componentwise convex functions, see \cite{Bachoc2022MaxMod} (Section SM1).
    Extending this equivalence to other sets of functions $\convex$  is an open problem.
    %Note that a similar result holds for the convexity and boundedness of $f$ when $D=1$, see \cite{maatouk2017gaussian}.
\end{remark}

\section{Sequential construction of constrained baGPs via MaxMod} \label{sec:MaxMod algorithm}

In the previous section, we built the predictor $\widehat{Y}^{S}_{\subpartition}$ defined in \eqref{eq:predictor}. This construction depends on the subdivisions $\subdivset=(s^{(1)},\ldots, s^{(D)})$, the subpartition $\subpartition$, and the convex set $\convex'$. Additionally, as discussed in Section \ref{subsec:Interpolation constraints and computation cost}, the computational cost of the predictor increases with the total number of basis function $\nknots$. This section provides an iterative methodology for optimally selecting the subpartition $\subpartition$ and the subdivisions $\subdivset$.

%\subsection{The MaxMod algorithm}\label{subsec:The Criteria and algorithm}

%We have detailed the method to construct the predictor from a given subpartition $\subpartition$ and subdivisions $\subdivset$. It is clear that the quality of the predictor will depend on these two parameters

The idea is to sequentially update, in a forward way, $\subpartition$ and $\subdivset$. To this purpose, we provide different choices to enrich $\subpartition$ and $\subdivset$ at each step of the sequential procedure:  activating a variable, refining an existing variable, merging two blocks. 

%defined from $(\subdivset,\subpartition)$, together with the criteria will determine how the predictor functions will be constructed, as our algorithm is greedy it is essential to be careful about the choices we have if we want our method to converge to the right construction. 

%\subsubsection{Construction of the set of choices to expand the basis}
%\paragraph{Choices to increment basis.}
%Let $\subdivset=\{s^{(1)},\ldots, s^{(\sizeblock)}\}$ and consider $\partition=\{\blockset_1,\ldots, \blockset_{\sizeblock}\}$ the three choices to update $(\subdivset,\subpartition)$ to $(\subdivset^\star,\subpartition^\star)$ are:

\subsection{Possible choices to update subpartition and the subdivisions}\label{subsec:choicesMaxMod}

To formalize the procedure, let us
write $\subdivset=(s^{(1)},\ldots, s^{(\sizeblock)})$ and  $\partition=\{\blockset_1,\ldots, \blockset_{\sizeblock}\}$.
Define the updated values of $\subpartition$ and $\subdivset$ after one of these three choices as $\Mstar = (\subdivset^\star, \subpartition^\star)$ with    $\subdivset^\star=(s^{\star(1)},\ldots, s^{\star(D)})$ and $\subpartition^\star=\{\blockset^\star_1, \ldots, \blockset^\star_\sizeblock\}$.
\begin{itemize} 

\item \textbf{ACTIVATE.} Activating a variable $i$ (for which $s^{(i)}=\emptyset$). Define $s^{\star(i)}:=(0,1)$,
$s^{\star(j)} = s^{(j)}$ for $j \neq i$,
and $\subpartition^\star :=\subpartition\cup \{i\}$.

\item \textbf{REFINE.} Refining an existing variable $i$ by adding a (one-dimensional) knot $t \in [0,1]$. We define 
\[
\subdivset^\star
:= (s^{(1)},\ldots, s^{(i-1)},\operatorname{ord}(s^{(i)}\cup t), s^{(i+1)},\ldots, s^{(D)}).
\]
Here, $\operatorname{ord}(\cdot)$ is an operator that sorts the knots in an increasing order. Assuming that $s^{(i)}_k<t<s^{(i)}_{k+1}$, then
$\operatorname{ord}(s^{(i)}\cup t)=(s^{(i)}_1, \ldots, s^{(i)}_k, t,s^{(i)}_{k+1}, \ldots, s^{(i)}_{m_i}).$ 
\item \textbf{MERGE.} Merging two blocks $\blockset_a$ and $\blockset_b$. We let $\subdivset^\star:=\subdivset$ and $\subpartition^\star:=
\left\{
\subpartition \backslash \{\blockset_a,  \blockset_b\}
,
\blockset_a \cup  \blockset_b
\right\}$.

%\left\{ 
%\blockset_a\cup \blockset_b , 
%\left\{ 
%\blockset_j
%\right\}_{j=1, j\notin \{a,b\}}^{B}
%\right\}
%$.
\end{itemize}
These options define a set 
\begin{equation*} \label{eq:Mstar}
\Mstar(\subdivset,\subpartition)=\{ (\subdivset^\star,\subpartition^\star)  \text{ that can be obtained from the three choices above starting from $(\subdivset,\subpartition)$}\}.
\end{equation*}
Now we define the MaxMod criterion in order to select a couple $(\subdivset^\star,\subpartition^\star)$ in $\Mstar(\subdivset,\subpartition)$.

\subsection{Construction of the MaxMod criterion}
\label{subsec:MaxModcriterion}

The MaxMod criterion combines two different subcriteria. The first one is the $\mathbf{L^2}$\textbf{-Modification (L2Mod)} criterion, defined between two estimators constructed from different subdivisions and subpartitions. This criterion has been used in the previous versions of MaxMod for dealing with non-additive and additive constrained GPs \cite{Bachoc2022MaxMod, lopez2022high}:
\begin{equation}\label{eq:L2criterion}
\Lcritfull =
\left\|\widehat{Y}^{\subdivset^\star}_{\subpartition^\star}-\widehat{Y}^\subdivset_{\subpartition}\right\|^2_{L^2} = \int_{[0,1]^D}\left(
\widehat{Y}^{\subdivset^\star}_{\subpartition^\star}(x)-\widehat{Y}^\subdivset_{\subpartition}(x)
\right)^2\, dx.
\end{equation}
Above,
$\widehat{Y}^\subdivset_{\subpartition}$ and $\widehat{Y}^{\subdivset^\star}_{\subpartition^\star}$ are the predictors constructed in~\eqref{eq:predictor}.
%, $\basis^{\subdivset}_{\subpartition}$ and $\basis^{\subdivset^\star}_{\subpartition^\star}$ are the bases spanning the finite vector spaces in which the finite-dimensional GPs have their realisations defined in~\eqref{eq:FinitDimBasisSpace}. 
%A closed-form expression of this criterion  is developed in Appendix \ref{appendix:squareNormCriteria}.
This criterion can be computed efficiently thanks to the following proposition  (see Appendix \ref{appendix:squareNormCriteria} for the proof).

\begin{proposition}[Closed form for the L2Mod criterion]\label{prop:algebraic expression}
    Let $\widehat{Y}^{\subdivset^\star}_{\subpartition^\star}$ and $\widehat{Y}^\subdivset_{\subpartition}$ be the two predictors defined in \eqref{eq:predictor}. Let $\setL^{\subdivset}_{\subpartition}$ and $\setL^{\subdivset^\star}_{\subpartition^\star}$ be the corresponding multi-indices sets defined in \eqref{eq:spaceESP:setLSP}. Then, with the vectors $\beeta \in \R^{|\setL^{\subdivset^\star}_{\subpartition^\star}|}$ of~\eqref{eq:eta vector}, $\bE \in\R^{|\setL^{\subdivset^\star}_{\subpartition^\star}|}$ of~\eqref{eq:E_ul} and the matrix $\bPsi \in M_{|\setL^{\star}|}(\R)$ defined in~\eqref{eq:Psi matrix}, we have the explicit expression:
    \begin{equation}
        \Lcritfull =
        \beeta^\top \bPsi \beeta + (\beeta^\top \bE)^2- \sum_{1\leq j \leq B}  \left(\beeta_j^\top \bE_j\right)^2.
    \end{equation}
    Furthermore, the matrix $\bPsi$ is sparse and the computational cost of $\Lcritfull$ is linear with respect to $|\setL^{\subdivset^\star}_{\subpartition^\star}|$.
\end{proposition}

Unlike the previous implementations of \cite{lopez2022high,Bachoc2022MaxMod}, which only quantify the difference between the two predictors, we aim to also account for improvements in prediction errors. Therefore, we measure the \textbf{Squared Error (SE)} criterion:
\begin{equation}\label{eq:SEcriterion}
\SEcrit(\subdivset^\star,\subpartition^\star) = \left\|\predictorstar(\bX)- \bY \right\|^2.
\end{equation}
Hence, we define the final selection criterion $\MaxModCrit$ of the MaxMod procedure as a combination of the two previous criteria:
\begin{equation}\label{eq:MaxModCriterion}
\MaxModCrit(\subdivset^\star,\subpartition^\star) = \frac{\Lcrit((\subdivset,\subpartition),(\subdivset^\star,\subpartition^\star))}{ (|\setLSPstar|-|\setLSP|)^{\alpha} \SEcrit(\subdivset^\star,\subpartition^\star)^{\gamma}}.
\end{equation}
Note that we also account for the difference of the bases sizes $|\setLSPstar|-|\setLSP|$, as our aim is to keep the dimension of the active space $E^\subdivset_\subpartition$ relatively low to have efficient computation over the predictors. The coefficients $\alpha >0, \gamma>0$ give flexibility to the MaxMod procedure. Large values of $\alpha$ lead to stronger penalties for merging blocks. %and lead too underestimate the size of the blocks 
%while lower values lead to consider bigger blocks.  
Larger values of $\gamma$ increase the importance of the $\SEcrit$ criterion.
We tried our method with $\alpha := (1,1.2,1.4)$ and $\gamma := (1, 0.5)$ over a range of test functions and the best results for recovering the blocks
were obtained with  $\alpha = 1.4$ and $\gamma = 0.5$. 
Thus we fix these values for the rest of the paper. 
Notice that~\eqref{eq:SEcriterion} is less reliable when data are noisy. Moreover, the SE can be very small, even when the predictor is not a good relative approximate, if the values of $\bY$ are themselves concentrated. As a stopping criterion for our algorithm, we consider the SE divided by the empirical variance $\var(\bY)$. This makes the stopping criterion invariant to rescaling of $\bY$. 
Algorithm~\ref{Algo:MaxMod} summarizes the implementation of MaxMod.

\begin{algorithm}[t!]
\caption{MaxMod}
\label{Algo:MaxMod}
\begin{algorithmic}[1]
\REQUIRE Observations $(\bX, \bY)$,  stopping criteria parameters
$\epsilon_1,\epsilon_2 \in (0,1)$, maximal number of iterations $M$
\ENSURE The subdivision $\subdivset$, the partition $\subpartition$ and the predictor $\widehat{Y}^{S}_{\subpartition}$
\STATE $\subdivset=((),\ldots,\,())$,
$\subpartition = \{\}$, $c_1=2\epsilon_1$, $c_2=2\epsilon_2$, $i=0$
\WHILE {$c_1>\epsilon_1$ and $c_2 >\epsilon_2$ and $i \leq M$}
\STATE $(\subdivset^\star,\subpartition^\star)= \arg\max_{(\subdivset',\subpartition') \in \Mstar(\subdivset,\subpartition)} \mathcal{K}((\subdivset,\subpartition),(\subdivset',\subpartition'))$ (see definition in \eqref{eq:MaxModCriterion})
\STATE $c_1=\Lcritfull$
\STATE $c_2=\SEcrit(\subdivset^\star,\subpartition^\star)/\var(\bY)$
\STATE $\subdivset = \subdivset^\star$, $\subpartition = \subpartition^\star$ 	
\STATE $i=i+1$
\ENDWHILE
\STATE Compute $\predictor$ according to 
%Section~\ref{subsec:Projection Finite Dimension GP}
\eqref{eq:predictor}
\RETURN $(\subdivset,\subpartition, \predictor)$
\end{algorithmic}
\end{algorithm}

\section{Numerical experiments}
\label{sec:Numerical results}

\subsection{General settings}

\paragraph{Numerical implementations.} The implementations of the bacGP framework and MaxMod have been integrated into the R package \texttt{lineqGPR}~\cite{LineqGPR}. Both the source codes and notebooks to reproduce some of the numerical illustrations presented in this section are available in the GitHub repository: \url{https://github.com/anfelopera/lineqGPR}. The experiments here have been executed on a 12th Gen Intel(R) Core(TM) i7-12700H processor with 16 GB of RAM.

To define the bacGP model, we consider tensorized Mat\'ern $5/2$ kernels (see Section~\ref{subsec:Unconstrained baGPs}). %Each block has one variance parameter $\sigma_j^2 \in \mathbb{R}^+$ and one length-scale parameter per dimension, denoted $\theta_{i} \in \mathbb{R}^+$. 
%This results in a total of $B+D$ covariance parameters where $B$ is the number of blocks.
We denote the set of covariance parameters as $\Theta = ((\sigma_1^2, (\theta_i)_{i\in \blockset_1}), \ldots, (\sigma_B^2, (\theta_{i})_{i \in \blockset_B}))$.
%With this setup, the kernel $k_j$ is given by
%\begin{equation}\label{eq:kernel}
%	\kernlj (\bx, \bx')=\sigma^2_j \prod_{i\in \blockset_j} \kernmat_{\theta_{i}}(x_i,x'_i),
%\end{equation}
%for every $\bx, \bx' \in \R^{|\blockset_j|}$, where functions $\kernmat_{\theta_{i}}$ are the normalized one-dimensional mat\'ern kernel:
%\[
%    \kernmat_\theta(x,x')= \left(1+\sqrt{5}\frac{|x-x'|}{\theta}+ \frac{5}{3} \frac{|x-x'|^2}{\theta^2}\right)\exp\left(-\sqrt{5} \frac{|x-x'|}{\theta}\right).
%\]
Both $\Theta$ and the noise variance $\tau^2$ are estimated via (multi-start) maximum likelihood (see Appendix~\ref{appendix:Kernel Hyperparameters selection} for a further discussion). 
%As discussed in Section~\ref{subsec:Computation of the predictor and algorithm cost}, we assume an additive Gaussian noise $\epsilon \sim \mathcal{N}(0, \tau^2)$, the variance $\tau^2$ is also an hyperparameter to be adjusted. 
The noise is required to ``relax'' the interpolation condition when modeling additive functions and to speed-up numerical computations. It also enhances numerical stability by preventing issues during the inversion of the covariance matrix defined in expression~\eqref{eq:Woodbury inversion}.
%The noise variance $\tau^2$ is also estimated via maximum likelihood. 

\paragraph{Training datasets.}  In the synthetic examples, as recommended by~\cite{lopez2022high} for additive constrained GPs, we consider training datasets based on random Latin hypercube designs (LHDs). While using LHDs is not required to perform the bacGP framework nor MaxMod, it is often recommended to promote more accurate predictions when dealing with additive functions~\cite{Stein1987LHS}. For the LHDs, we choose a design size $n = k \times D$, with $D \in \mathbb{N}$ the dimension of the input space, and $k \in \mathbb{N}$ a multiplication factor that can be arbitrarily chosen.
Setting $k < 10$ is often considered reliable when accounting for additional information provided by additive structures or inequality constraints within GP frameworks~\cite{lopez2022high}. In our study, we fix $k = 3$ when focusing on the assessment of predictions. This value is set based on the maximal number of covariance parameters to be estimated, which is $2D + 1$ for an additive process that neglects interactions between variables (worst case). For testing MaxMod's ability to identify the partition $\subpartition$, we manually set $k = 7$, which provides stable inference results.

\paragraph{Performance indicators.} We assess the quality of predictions in terms of the $Q^2$ criterion computed from the Standardized Mean Square Error (SMSE) as
\begin{equation}
Q^2
= 1 - \SMSE(y, \widehat{y})
= 1 - \frac{\sum_{i=1}^{n}(y_i-\widehat{y}_i)^2}{\sum_{i=1}^{n} (y_i-\overline{y})^2},
\label{eq:Q2}
\end{equation}
where $(y_i)$ are the observations, $(\widehat{y}_i)$ are the corresponding predictions, and $\overline{y} = \frac{1}{n} \sum_{i=1}^{n} y_i$ is the empirical mean. The $Q^2$ criterion is equal to $1$ if predictions exactly coincide with observations, and is smaller otherwise. In the synthetic examples, where the target function can be freely evaluated, the $Q^2$ is computed via Monte Carlo using $10^5$ points from a maximin LHD. For the coastal flooding application, it is computed only on the subset of the dataset that is not used for training the models.

In the coastal application, to ensure comparability with previous models tested on the same application, we also consider the bending energy criterion given by
\begin{equation}
    E_n(y, \widehat{y})=\frac{\sum_{i=1}^{n} (y_i -\widehat{y}_i)^2}{\sum_{i=1}^{n} y_i^2}.
    \label{eq:En}
\end{equation}

\subsection{Monotonicity in high dimension}
For testing the bacGP in high dimension, we consider the non-decreasing block-additive target function $y:[0,1]^D\to \R$:
\begin{equation}\label{eq:multidimFun}
y(\bx)=\sum_{j=1}^{D/2} \arctan\left(5 \left[ 1 -\frac{j}{d+1}\right](x_{2j-1}+2x_{2j})\right).
\end{equation}
The structure of $y$ is inspired by the additive functions studied in \cite{Bachoc2022MaxMod,lopez2022high}, but allowing interactions between input variables. More precisely, we consider $D/2$ blocks composed by non-overlapping pairs of input variables. A scale factor that varies with $j \geq 1$ is introduced to control the growth rate of a given block. As observed in~\eqref{eq:multidimFun}, this growth rate decreases as the index $j$ increases.

For different values of $D \geq 10$, we assess baGP models with and without non-decreasing constraints. The focus here is to compare the quality of bacGP predictors with respect to the unconstrained baGP predictor. For the bacGPs, we set 6 knots uniformly distributed over each variable as subdivisions and the partition $\subpartition=\{\{2j-1,2j\}, 1 \leq j \leq D/2 \}$. We denote the MAP estimator in~\eqref{eq:predictor} as the \textit{bacGP mode} and the estimator obtained by averaging Monte Carlo samples as the \textit{bacGP mean}. For the latter, we use the exact Hamiltonian Monte Carlo (HMC) sampler proposed by~\cite{Pakman2014Hamiltonian}. The unconstrained GP estimator is referred here as the \textit{GP mean}.

\begin{table}[t!]
\centering
\caption{Results (mean $\pm$ one standard deviation over ten replicates) on the monotonic example in~\eqref{eq:multidimFun} with $n=3D$. Both computational cost and quality of the bacGP predictions (mode and mean) are assessed. For the computation of the bacGP mean, different number of HMC samples are used and they are indicated as $N_{sim}$. Due to computational overhead, $N_{\text{sim}}$ decreases as $m$ increases.}
\label{tab:monotonicityHD} 
\begin{tabular}{cccccccc}
\toprule			
\multirow{2}{*}{$D$}  & \multirow{2}{*}{$m$} & \multirow{2}{*}{$N_{sim}$} & \multicolumn{2}{c}{CPU Time $[s]$} & \multicolumn{3}{c}{$Q^2$ [\%]} \\
& & & bacGP mode & bacGP mean & baGP mean & bacGP mode & bacGP mean \\	
\midrule						
10   & 180  & $10^4$ & 0.27  $\pm$ 0.01  & 15.66 $\pm$ 3.04 					& 78.9 $\pm$ 9.0 & 89.5 $\pm$ 4.5 & \textbf{90.4 $\pm$ 3.1} \\

20  & 360  & $10^3$ & 0.50  $\pm$ 0.05    & 10.78 $\pm$ 1.61 					& 82.5 $\pm$ 4.7 & 92.0 $\pm$ 1.0  & \textbf{92.7 $\pm$ 0.1} \\

40  & 720 & $10^2$ & 1.09 $\pm$ 0.06    & 6.46 $\pm$ 0.67   			& 86.1 $\pm$ 1.8 & \textbf{91.3 $\pm$ 1.2}  & 90.6 $\pm$ 1.8 \\

80  & 1440 & $10^2$ & 3.17 $\pm$ 0.15   & 17.82 $\pm$ 5.12 	& 86.1 $\pm$ 1.3 & \textbf{92.0 $\pm$ 1.0}  & 91.6 $\pm$ 1.2 \\

120 & 2160 & $10^2$ & 6.47 $\pm$ 0.35 & 47.68 $\pm$ 4.93 & 87.4 $\pm$ 1.1  & \textbf{89.9 $\pm$ 0.8} & 87.4 $\pm$ 0.7\\ % 100 HMC samples
\bottomrule
\end{tabular}
\end{table}
Table~\ref{tab:monotonicityHD} presents the CPU times and $Q^2$ values of the GP predictors averaged over 10 replicates using different random LHDs with size $n = 3D$. We observe an overall improvement in prediction accuracy when constraints are incorporated, resulting in $Q^2$ increases ranging between $2.5\%$ and $11\%$. Particularly, the predictor based on the bacGP mode often outperforms others while maintaining computational tractability. We also note that bacGP mean leads to competitive $Q^2$ values but requires more computationally intensive implementations. Lastly, as the number of observations increases, we notice that the inequality constraints are learned from the training data in the unconstrained baGP. Hence, the use of the constrained model is more advantageous in applications where data is scarce.

\subsection{Model selection via MaxMod}
\label{subsec:MaxMod}
We now consider the following 6D function aiming to test the efficiency of MaxMod:
\begin{equation}\label{eq:toolexample2Dblockadditiv}
y(\bx)= 2x_1x_3 + \sin(x_2x_4) + \arctan(3x_5+5x_6).
\end{equation}
It is worth noting that $y$ is non-decreasing with respect to all its input variables.
A prior sensitivity analysis suggests that MaxMod is likely to prioritize activating the first input variables, given their higher contribution to the Sobol indices: $S_1=S_3 \approx 0.41$, $S_2=S_4\approx 0.08$, $S_5\approx 0.05$, $S_6\approx 0.1$. Furthermore, since the function $(x_1,x_3)\mapsto x_1x_3$ is componentwise linear, we anticipate the algorithm to activate and merge only these variables, without any further refinement. Similarly, functions defined over other variables may not belong to the vector space spanned by the tensorized hat basis functions, suggesting that more knots in those subdivisions might be necessary. To demonstrate that MaxMod is also effective in dimension reduction, we slightly modified the function $y$ by introducing twenty additional dummy input variables, denoted as $x_7, \ldots, x_{26}$. Under this scenario, we expect the algorithm to focus on activating the first six input variables. 

\begin{figure}[t!]
\centering
\begin{tabular}{cc}
\includegraphics[height=0.45\linewidth]{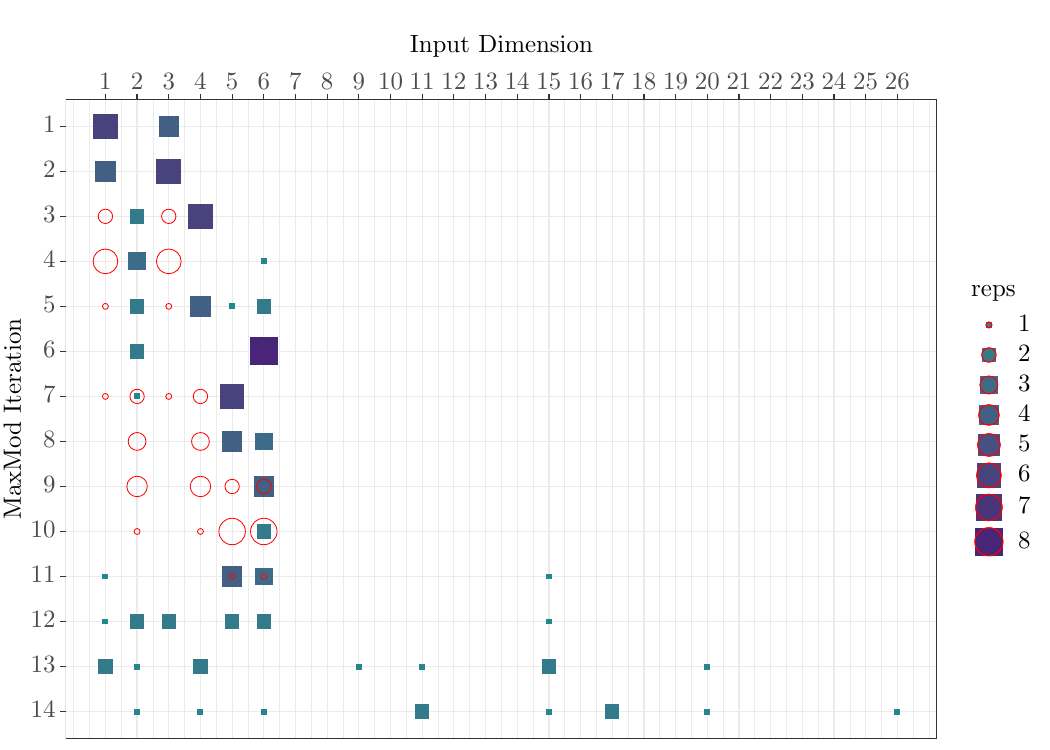}
\includegraphics[height=0.4275\linewidth]{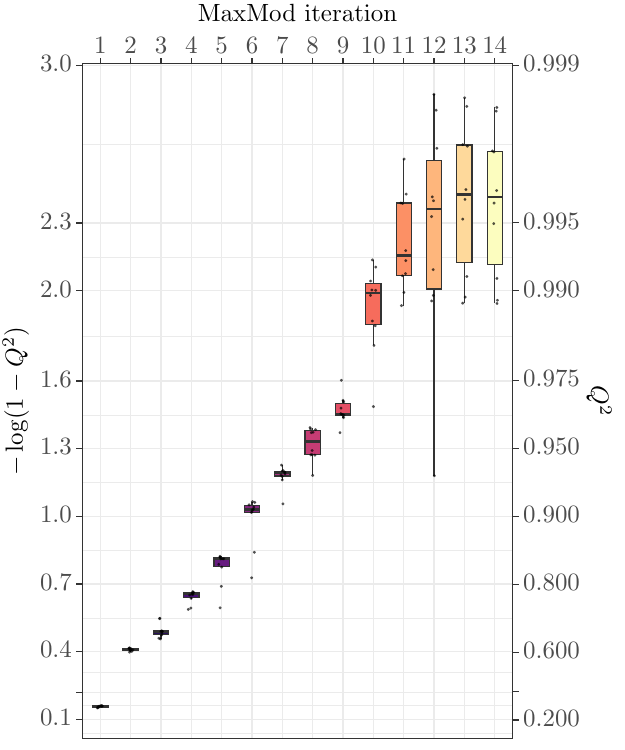}
\end{tabular}
\caption{Model selection via MaxMod when considering the target function in~\eqref{eq:toolexample2Dblockadditiv}. The panels show: (left) the choices made by MaxMod and (right) the boxplot of the $Q^2$ criterion per iteration of the algorithm. Results are shown for ten replicates of the experiment considering different LHD-based training datasets with $n = 7D_o$ with $D_o = 6$ the number of active input variables. In the left panel, squares represent variables that are newly selected or refined by MaxMod, while red circles represent variables that are being merged. The size and color of the markers indicate the frequency of the corresponding choice made by MaxMod over multiple iterations. In the right panel, colours of the boxes correspond to the median: lighter colours correspond to higher medians.}
\label{fig:toolexChoicesQ2}
\end{figure}
Figure~\ref{fig:toolexChoicesQ2} shows the decisions made by MaxMod alongside boxplots showing the $Q^2$ criterion per iteration for the ten different replicates. In the right panel, we observe that MaxMod initially activates variables $x_1$ and $x_3$. Following their activation in the first two iterations, the algorithm then decides between merging them into a single block or activating variables $x_2$ and $x_4$. It then proceeds with options such as creating a new block with $x_2$ and $x_4$, activating and merging the remaining variables $x_5$ and $x_6$, or refining the knots of $x_2, x_4, x_5$ and $x_6$. As anticipated, variables $x_1$ and $x_3$ are less frequently refined. By the  11th iteration, the algorithm consistently identifies the true partition $\partition = \{\{1,3\}, \{2,4\}, \{5,6\}\}$ across all ten replicates in the experiment. In the right panel, we observe that the $Q^2$ criterion improves with each iteration, leading to a stable (median) behavior above $Q^2 = 0.995$ after twelve iterations. 

Note that beyond twelve iterations, MaxMod starts considering the activation of dummy variables or refining variables $x_1$ and $x_3$, which is an undesired behavior considering the nature of the target function in~\eqref{eq:toolexample2Dblockadditiv}. This behavior may be attributed to significant empirical correlations between dummy variables and active variables due to the experimental design. To mitigate this issue, adapting the stopping criterion of the algorithm to achieve convergence earlier when neither the $\Lcrit$ nor the $\SMSE$ criterion shows significant improvement would be beneficial.
%The next figure will provide some insights over the MaxMod algorithm as it focus on the value of the criteria introduced in section~\ref{sec:MaxMod algorithm}.

\begin{figure}[t!]
    \centering
    %\begin{minipage}{0.63\linewidth}
    {\includegraphics[width=0.9\linewidth]{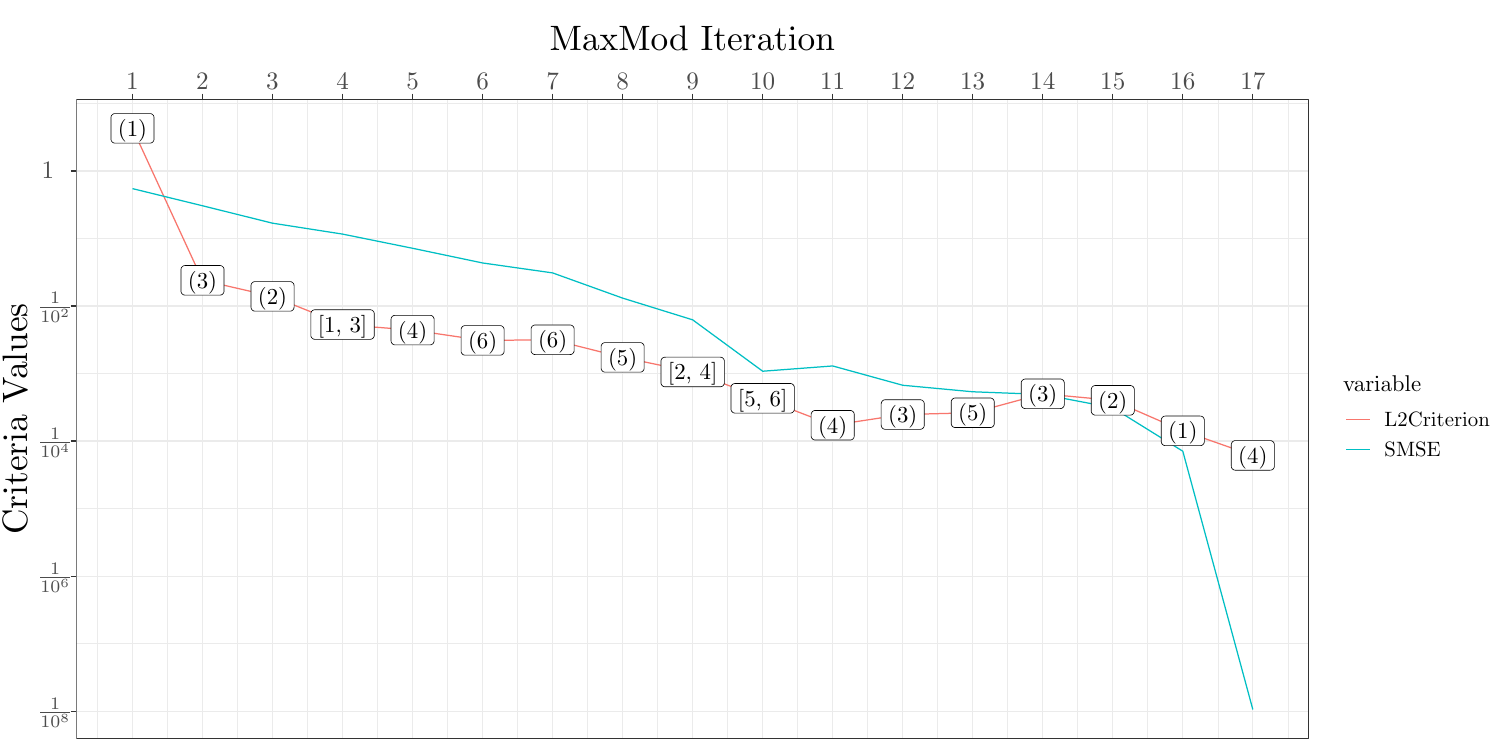}}
    %\end{minipage}
    %\hspace{0.3ex}
    %\begin{minipage}{0.34\linewidth}
    %	\vspace{2.5ex}
    \caption{Evolution over MaxMod iterations of the $\Lcrit$ (red) and $\SMSE$ % := \SEcrit / \sum_{i=1}^n(y_i-\bar{y})^2$ 
    (blue). Both criteria are defined in~\eqref{eq:L2criterion} and~\eqref{eq:Q2}, respectively. The choice made by the algorithm per iteration is displayed in a text box where ``$(i)$" indicates the activation or the refinement of the variable $i$, while ``$[i_1,\ldots, i_k]$" indicates the creation of a block composed of variables $i_1,\ldots,i_k$.}
    \label{fig:toolCriteria}
    %\end{minipage}		
\end{figure}
Figure~\ref{fig:toolCriteria} shows that both $\SEcrit$ and $\Lcrit$ tend to decrease over iterations for a fixed replicate. On the 13th iteration of MaxMod, it can be observed that the SMSE slightly increases. One might wonder why the interpolation of the predictor $\predictor$ on the 10th iteration is better than the one in the 11th. Indeed, $\predictor$ minimizes an interpolation problem in an RKHS \cite{wahba1990spline}. In fact, we have an inclusion of RKHS as for the bases. Hence, it is expected that the solution in a higher-dimensional space would better interpolate the observations. However, the noise variance $\tau^2$ alters the nature of the optimization problem. This issue is further discussed in Appendix~\ref{appendix:SEevolution} where we theoretically demonstrate that increasing the dimensionality of the RKHS space can degrade the solution in terms of interpolation.

\subsection{Real application: Coastal flooding}
\label{subsec:coastalAPP}

We now examine a coastal flood application in 5D previously studied in~\cite{azzimonti2019profile,LopezLopera2019lineqGPNoise,Bachoc2022MaxMod}. The dataset is available in the R package \texttt{profExtrema}~\cite{azzimonti2018profextrema}. The application focuses on the Boucholeurs district located on the French Atlantic Coast near the city of La Rochelle. This site was hit by the Xynthia storm in February 2010, which caused the inundation of several areas and severe human and economic damage.
%The interest in this database arises from a historical flood event at the Boucholeurs area in \textit{La Rochelle}, France, triggered by an overflow from the Atlantic Ocean during the Xynthia storm in 2010. To prevent future coastal flood damages like those caused by the Xynthia storm, reliable forecasting and early-warning systems are crucial.
We analyze here the flooded area ($A_{flood}$ [$m^2$]) induced by overflow processes by using the hydrodynamic numerical model detailed in \cite{rohmer2018casting}. The dataset comprises $200$ numerical results of $A_{flood}$, each of them being related to the values of the parameters that describe the temporal evolution of the tide and the surge. The tide temporal signal is simplified and assumed to be represented by a sinusoidal signal, parameterized by the high-tide level $T > 0$. The surge signal is modeled as a triangular function defined by four parameters: the surge peak $S > 0$, the phase difference ($\phi$ [$h$]) between surge peak and high tide, the rising time ($t^-$ [$h$]) and falling time ($t^+$ [$h$]) of the triangular signal. Figure \ref{fig:surgeTide_graphANOVA} (left panel) shows a schematic representation of both signals. We refer to~\cite{azzimonti2019profile,rohmer2018casting} for further details on the context and the physical meaning of these variables.

We assume, as illustrated by~\cite{azzimonti2019profile}, that $A_{flood}$ is non-decreasing with respect to $T$ and $S$. This assumption makes sense from the viewpoint of the flooding processes because both variables have a direct increasing influence on the offshore forcing conditions. In other words, the higher $T$ or $S$, the higher the total sea level (which is given by the sum of the tide and surge signals, see Figure~\ref{fig:surgeTide_graphANOVA}, bottom-left panel), and thus the higher the expected total flooded area. Adopting the procedure used in~\cite{LopezLopera2019lineqGPNoise}, we consider as outcome $y := \log_{10}(A_{flood})$ to ensure positivity, and we apply the transform $\phi \mapsto (1+\cos(2\pi \phi))/2$. There, these transformations led to improvements in the $Q^2$ criterion.

%\subsubsection{MaxMod with and without constraints}
We perform twenty replicates of the experiment using different training datasets, each comprising $35\%$ of the database (i.e., $n = 70$), and evaluate the $Q^2$ criterion on the remaining data. We propose a bacGP with non-decreasing constraints on the input variables $T$ and $S$. To prevent overfitting and ensure stable results, we early stop MaxMod after ten iterations, noting that stability is achieved after the first eight iterations (see Figure~\ref{fig:MaxModEnergyBRGM}). For prediction purposes, we focus solely on the predictor provided by MaxMod (see Algorithm~\ref{Algo:MaxMod}). We compare the $Q^2$ results to those obtained by the conditional mean of a non-additive unconstrained GP accounting only for the input variables already activated by MaxMod. The unconstrained model is implemented using the R package \texttt{DiceKriging}~\cite{RoustantDice2012}.

\begin{figure}[t!]
\centering
\begin{tabular}{cc}
\includegraphics[height=0.4\linewidth]{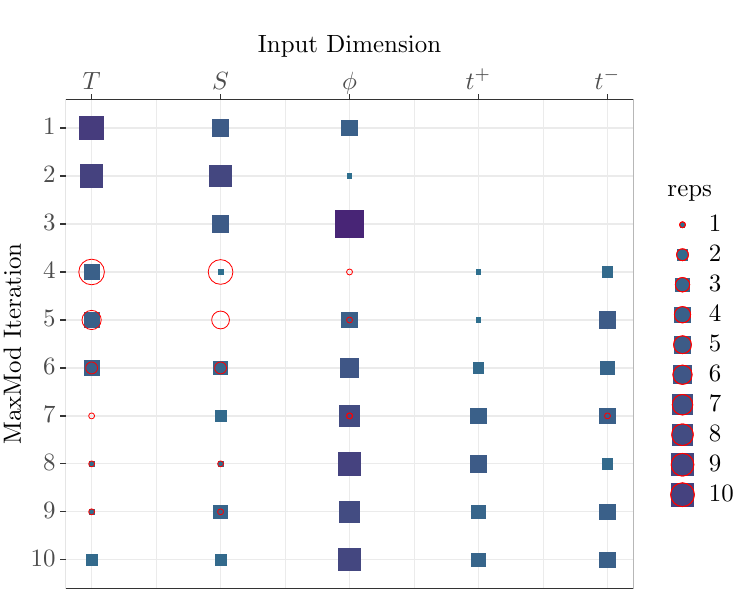}		
\includegraphics[height=0.373\linewidth]{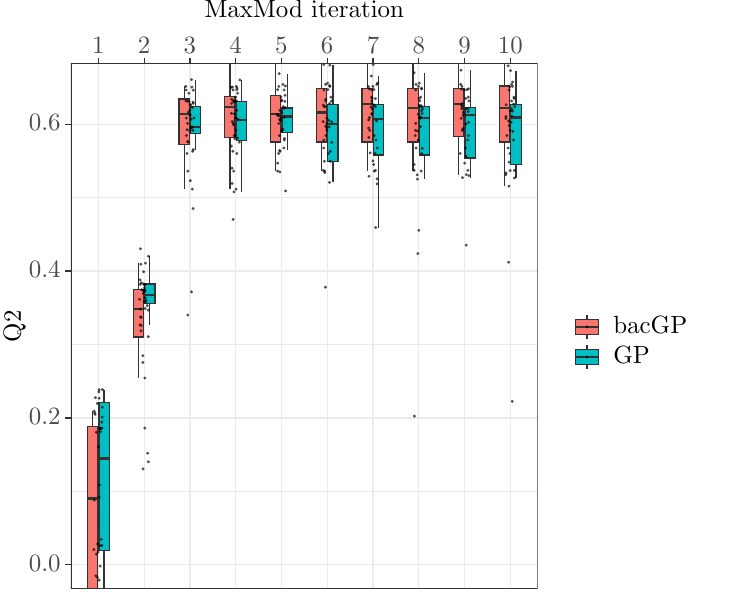}
\end{tabular}
\caption{Model selection via MaxMod for the coastal flooding application in Section~\ref{subsec:coastalAPP}. The panels show: (left) the choices made by MaxMod and (right) the boxplot of the $Q^2$ criterion per iteration of the algorithm. Results are shown for twenty replicates of the experiment considering different training datasets considering 35\% of the database (i.e. $n = 70$). Description of the first panel is the same as in Figure~\ref{fig:toolexChoicesQ2}. For the second panel, $Q^2$ results led by the bacGP (red) are compared to those obtained by an non-additive unconstrained GP (blue) defined on the active dimensions in the subpartition found by MaxMod. We recall here that $\bx = (T,S,\phi, t^+, t^-)$.}
\label{fig:Q2BRGM}
\end{figure}
Figure~\ref{fig:Q2BRGM} (left panel) illustrates the progression of MaxMod. Initially, it activates the variables $T$, $S$, and $\phi$. By the third iteration, it focuses on refining variables $T$ and $\phi$, merging $T$ and $S$, or activating additional variables. After ten iterations, the algorithm deems all five input dimensions relevant and suggests considering interactions between $T$ and $S$. %It frequently refines the variable $\phi$, stressing its importance.
These results can be interpreted in terms of flood processes:
\begin{itemize}
    \item \textbf{The importance of $S$ and $T$} is physically significant since these two variables have a direct impact on the sea level at the coast, and therefore on the total amount of water that can potentially invade inland in the event of flooding. \item \textbf{The mirroring role of $S$ and $T$} explains the relevance of merging them, i.e., the increase in $T$ or $S$ is interchangeable.
    \item \textbf{The importance of $\phi$} is natural if we consider that when it is equal to zero (i.e.,  when $(1+\cos(2\pi \phi))/2$ is equal to one), the tide and surge signals are in phase, and then the total sea level is maximum.
\end{itemize}
These interpretations are consistent with a sensitivity analysis using the FANOVA-decomposition in~\cite{muehlenstaedt2012} and the total interaction index in~\cite{fruth2014total}. The latter are estimated using a non-additive unconstrained GP (with a constant trend) trained on the entire database, and the estimator in~\cite{liu2006estimating} with 50$k$ function evaluations. Consistently with our results, this experiment highlights the relevant interaction between $S$ and $T$ and, to some extent, between $T$ and $\phi$, with a total interaction index of the order of 20\%. The interaction structure is shown in Figure~\ref{fig:surgeTide_graphANOVA} (right panel). 
\begin{figure}[t!]
    \centering
    \begin{minipage}{0.4\linewidth}
        \centering
        \includegraphics[height=0.13\textheight]{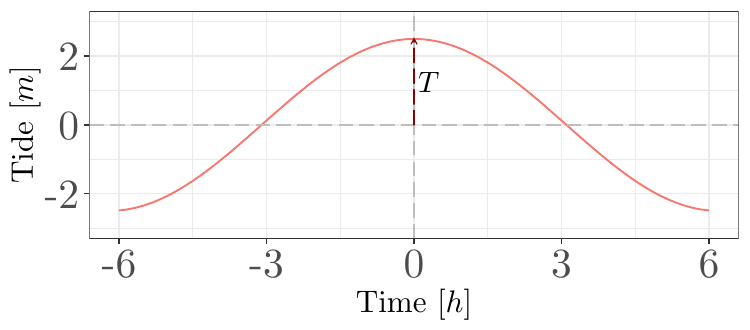}\\
        \vspace{-5.3ex}
        \includegraphics[height=0.13\textheight]{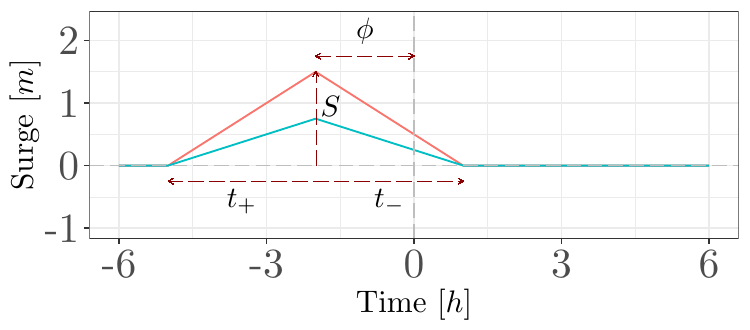}\\
        \vspace{-5.3ex}
        \includegraphics[height=0.13\textheight]{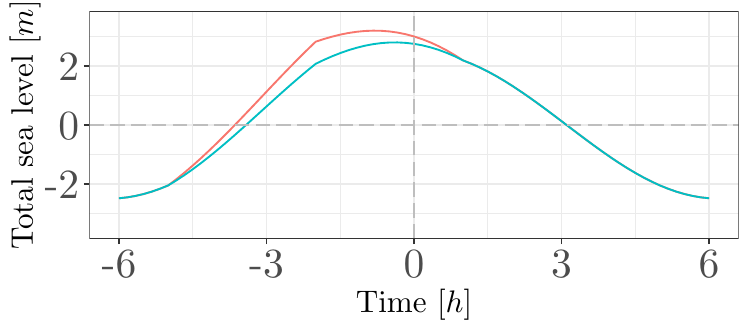}
    \end{minipage}
    \hspace{3ex}
    \begin{minipage}{0.37\linewidth}
        \centering
        {\includegraphics[height=0.28\textheight]{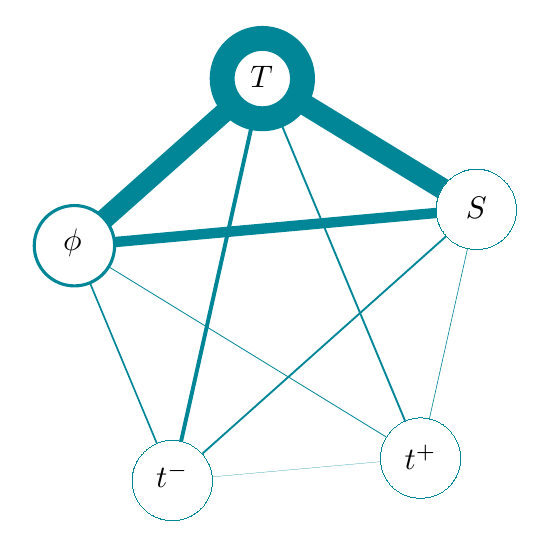}}
    \end{minipage}
    %\vspace{-3ex}
    
    \caption{(top-left) Schematic representation of the tide and (middle-left) and surge temporal signal used in the real test case. The input variables correspond to the parameters outlined with dashed arrows. The sum of both signals results in the total sea level (bottom-left) which determines the offshore forcing conditions of the hydrodynamic numerical model used to simulate the flooding processes. Two examples are shown where the surge peak $S$ varies. It can be observed that a lower value of $S$ corresponds to a lower total sea level height. (right) FANOVA graph representing the interaction structure in the coastal flooding application. The linewidth of the nodes is proportional to the Sobol first order index (main effect) and the linewidth of the  graph edges proportional to the total interaction index.}
    \label{fig:surgeTide_graphANOVA}
\end{figure}
\begin{figure}[t!]
    \centering
    \includegraphics[width=0.48\linewidth]{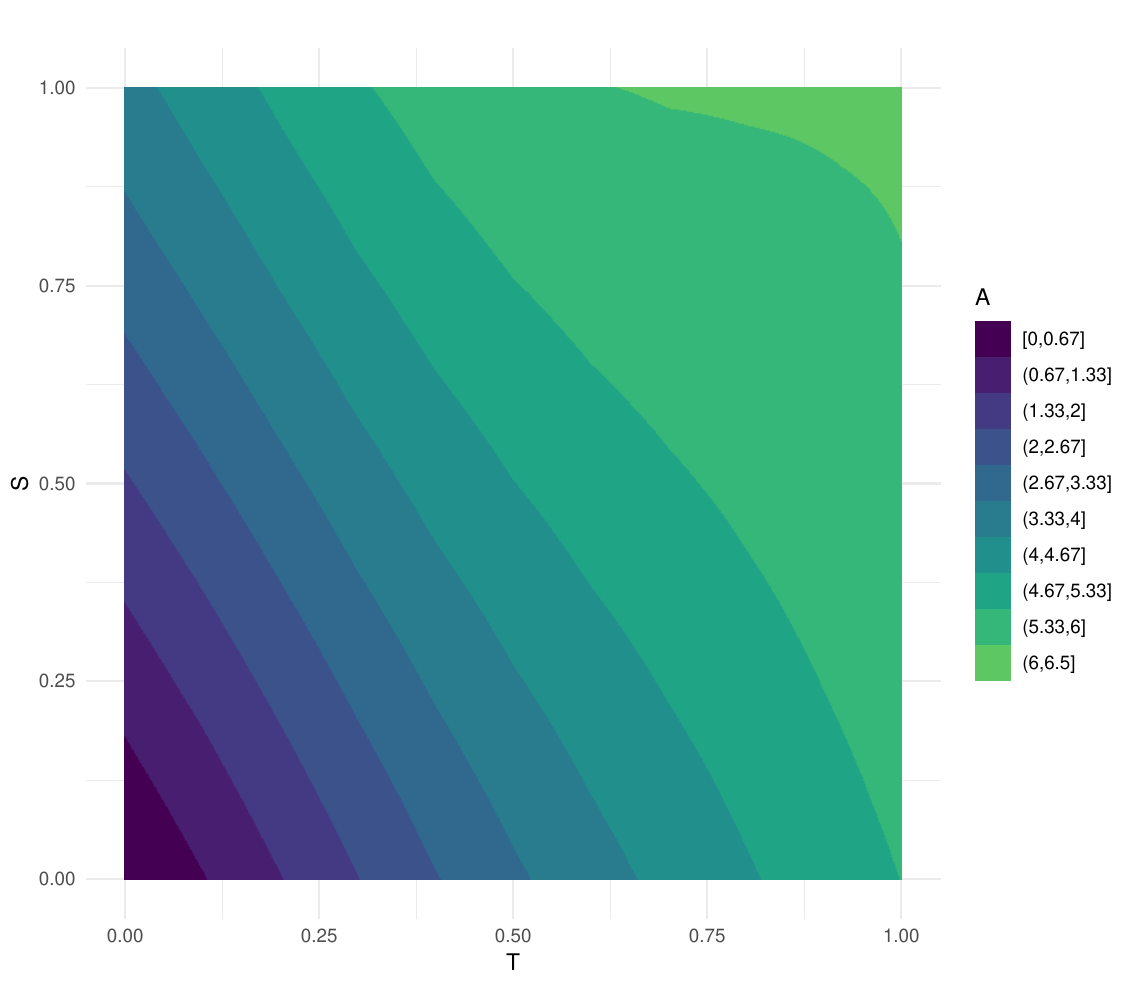}
    \includegraphics[width=0.48\linewidth]{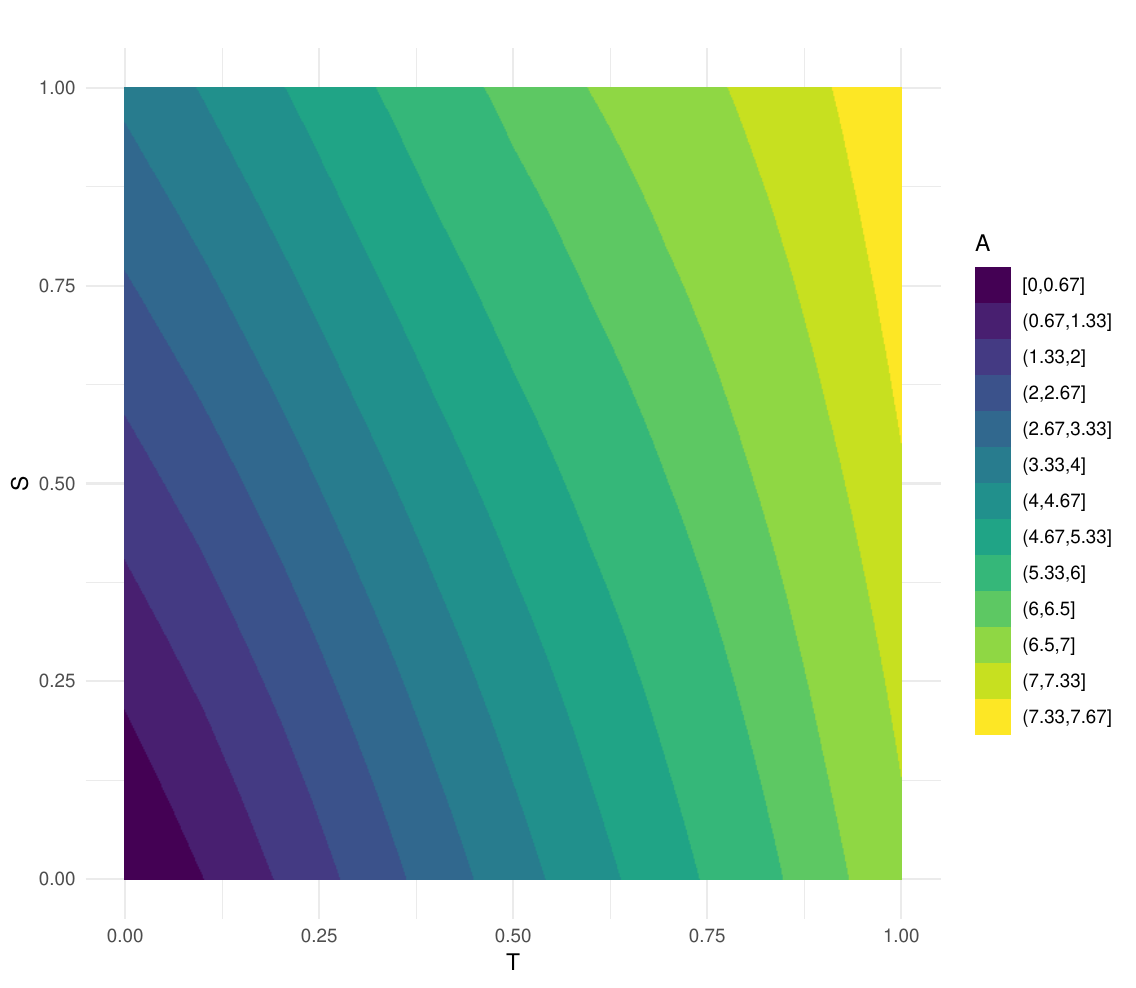}
    \label{fig:2Dplotcostalfun1}
    \caption{Bivariate representation of $\widehat{y}_1(S, T, \phi)$ for (left panel) $\phi = \pi$ and (right panel) $\phi = 0$ for the coastal flooding study in Section~\ref{subsec:coastalAPP}.}
    \label{fig:BRGM_2dplot}
\end{figure}

We recall that the aforementioned sensitivity analysis is obtained using the entire database (i.e. $n = 200$). Interestingly, when repeating the experiment with only 35\% of the database (i.e. $n = 70$), as suggested when testing MaxMod, the interaction structure could hardly be retrieved. The total interaction indices were highly variable over 20 replicates of the experiments. For $S$ and $T$, the total interaction index ranged between from 2\% to 15\%, and for $T$ and $\phi$, from 7\% to 23\%. On the other hand, MaxMod successfully identified the interaction between $S$ and $T$ even with the limited number of samples, although detecting the interaction $T$ and $\phi$ remained challenging.

The MaxMod algorithm has also the practical advantage of providing the functional relationships between the three variables, which allows us to get deeper insight in their joint influence on $A_{flood}$. We recall that, after convergence of MaxMod, the inferred additive structure of the function is $\widehat{y}(S,T,\phi, t_{+},t_{-})= \widehat{y}_1(S,T,\phi) + \widehat{y}_2(t_{+}) + \widehat{y}_3(t_{-})$. Figure \ref{fig:BRGM_2dplot}, showing 2-dimensional visualization of the function $\widehat{y}_1(S,T,\phi)$  for $\phi = \pi$ and $\phi = 0$, confirms an expected behaviour from the viewpoint of flood processes. It seems that $\widehat{y}_1(S, T, 0)> \widehat{y}_1(S, T, \pi)$ for any $(S,T)$, this observation aligns with the tide and surge signals becoming increasingly in phase, leading to a larger flooded area. For instance, consider the combinations of ($T$, $S$) for which the log-transformed flooded area $y := \log_{10}(A_{flood}) \geq 6$. As $\phi$ decreases, the range of admissible ($T$, $S$) combinations expands. This zone extends with respect to $\phi$ almost twice as rapidly along the $T$ axis compared to the $S$ axis. This suggests that exceedance is allowed for a more restrictive range of $S$ values, further confirming the dominant influence of $T$ (see Figure \ref{fig:surgeTide_graphANOVA}, right panel). A finer analysis is made in Appendix \ref{appendix:Analysis for the coastal flooding application}.

Regarding the $Q^2$ criterion (Figure~\ref{fig:Q2BRGM}, right panel), the first three iterations of MaxMod are crucial for activating the most ``expressive'' input variables, leading to $Q^2 > 0.6$. Subsequent iterations yield slight but consistent improvements, outperforming predictions from unconstrained GPs. Similar $Q^2$ results have been reported in~\cite{LopezLopera2019lineqGPNoise} for non-additive constrained GPs without MaxMod.

\begin{figure}
	\centering

\includegraphics[width=0.96\linewidth]{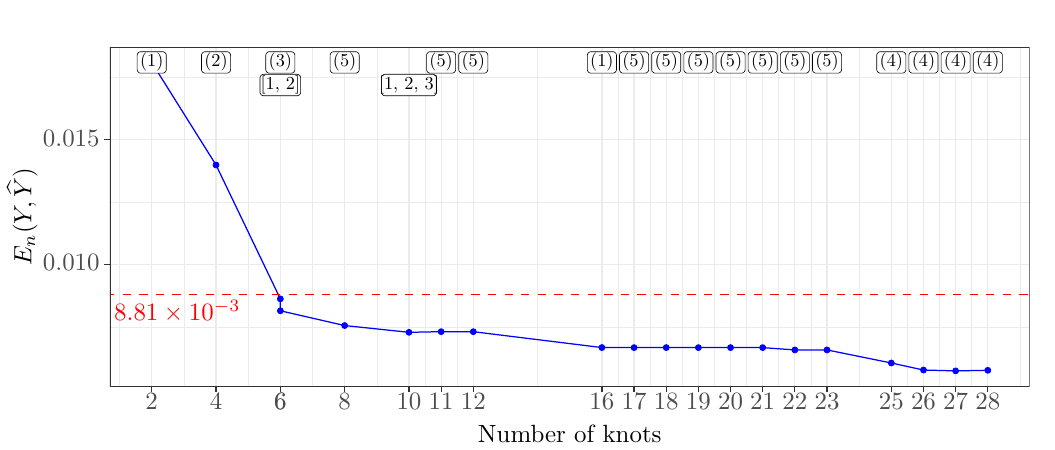}
\caption{Evolution of the bending energy $E_n$ through MaxMod iterations for the bacGP. The choices of the algorithm are detailed by labels defining the same choices of those described in Figure~\ref{fig:toolexChoicesQ2}. The red dashed line indicates the bending energy for the non-additive constrained GP after convergence of MaxMod~\cite{Bachoc2022MaxMod}.}
\label{fig:MaxModEnergyBRGM}	
\end{figure}

Finally, we aim to compare the results obtained here with those reported in~\cite{Bachoc2022MaxMod}. In their study, model selection for a non-additive constrained GP via MaxMod using the entire database resulted in a bending energy (see definition in \eqref{eq:En}) $E_n = 8.81\times10^{-3}$ for a model with $\nknots = 432$ multi-dimensional knots. Here, as shown in Figure~\ref{fig:MaxModEnergyBRGM}, comparable $E_n$ values are achieved after only three iterations of MaxMod, requiring significantly fewer knots $\nknots$. %The bending energy consistently decreases as the number of knots increases over the iterations.
After convergence, our framework attains $E_n(Y, \widehat{Y}) = 4.2\times10^{-3}$ with only $\nknots = 28$. This represents a substantial computational improvement in simulation tasks, as the complexity depends on sampling a $\nknots$-dimensional truncated Gaussian vector. The improvements in $E_n$ stem from exploiting the latent block-additive structure and incorporating the new criterion $\MaxModCrit$ (see definition in~\eqref{eq:MaxModCriterion}), which explicitly targets minimization of the squared error (equivalent to bending energy up to renormalization).

\section{Conclusion}
\label{sec:conclusions}

We introduced a novel block-additive constrained GP framework that allows for interactions between input variables while ensuring monotonicity constraints. As shown in the numerical experiments, the block-additive structure of the model makes it particularly well-suited for functions characterized by strong inter-variable dependencies, all while maintaining tractable computations. For model selection (i.e., the choice of the blocks), we developed the sequential MaxMod algorithm which relies on the maximization of a criterion constructed with the square norm of the modification of the MAP predictor between consecutive iterations and the square error of the  predictor at the observations. MaxMod also seeks to identify the most influential input variables, making it efficient for dimension reduction. Our approach provides efficient implementations based on new theoretical results (in particular the conditions for inclusion relationships between bases composed of hat basis functions and the corresponding change-of-basis matrices) and matrix inversion properties. R codes were integrated into the open-source library \texttt{lineqGPR}~\cite{LineqGPR}.

Through various toy numerical examples, we demonstrated the framework's scalability up to 120 dimensions and its ability to identify suitable blocks of interacting variables. We also assessed the model in a real-world 5D coastal flooding application. In the latter, the derived blocks together with the block-predictors have proven to be a key for interpreting the physical processes acting during flooding. Only a limited budget of observations of the coastal flooding simulator is necessary, which is beneficial given the high cost of this simulator, to identify the most influential factors, their interactions as well as their functional relationships.

The proposed work focused on applications satisfying monotonicity constraints, but it can be used to handle other types of constraints, such as componentwise convexity.
We note that many applications in fields such as biology and environmental sciences require handling boundedness and positivity constraints. For the additive case, these types of constraints do not verify the equivalence in~\eqref{eq:EquivalenceCCp}. Therefore, a potential future direction is to adapt the proposed framework to handle these constraints. 

Additionally, theoretical guarantees of the MaxMod algorithm could be further investigated. Indeed, it would be beneficial to show as in \cite{Bachoc2022MaxMod} that the sequence of predictors converges to the infinite-dimensional constrained minimization solution in the RKHS induced by the kernel of the GP as developed in \cite{Bay2016KimeldorfWahba,grammont2024error}. 
Moreover, except from \cite{Bachoc2010cMLE}, very few asymptotic results exist in the setting where the number of observations goes to infinity. It would be interesting to obtain more of these results, in particular related to the estimation of block structures.

\section*{Acknowledgement}
This work was supported by the projects GAP (ANR-21-CE40-0007) and BOLD (ANR-19-CE23-0026), both projects funded by the French National Research Agency (ANR). Research visits of AFLL at IMT and MD at UPHF have been funded by the project GAP and the National Institute for Mathematical Sciences and Interactions (INSMI, CNRS), as part of the PEPS JCJC 2023 call. We thank the consortium in Applied Mathematics CIROQUO, gathering partners in technological research and academia in the development of advanced methods for Computer Experiments, for the scientific exchanges allowing to enrich the quality of the contributions.
%Implementations in this paper are available in the R package \texttt{lineqGPR}~\cite{LineqGPR}. 
We finally thank Louis B\'ethune for the first Python developments of baGPs.

\renewcommand{\arraystretch}{1.5}
\begin{table}
\caption{List of symbols for Sections \ref{sec:Construction of the predictor} and \ref{sec:Conditioning a finite-dimensional baGP}, with the page numbers where the symbols are introduced.}
%		\color{black}
\centering
\small
\begin{tabular}{p{3.2cm}|p{10.5cm}|p{0.9cm}} 
\hline 
{\bf Symbol} & {\bf Description}  &	 {\bf Page} \\
\hline
$\subpartition=\{\blockset_1,\ldots,\blockset_B\}$ & Subpartition of $\{1, \dots, n\}$ &  \pageref{page:subpartition} \\
\hline 
$\blockset_j$ & Subset of $\{1, \dots, n\}$ corresponding to a block of variables & \pageref{page:blockset}\\
\hline 
$\sizeblock$ & Number of blocks of variables & \pageref{page:sizeblock}\\
\hline
$s^{(i)}=(t^{(i)}_1,\ldots, t^{(i)}_{m^{(i)}})$ & Subdivision for the variable $i$ (set of knots) & \pageref{page:subdivision}\\
\hline
$m^{(i)}$ & Size of the subdivision $s^{(i)}$ (number of one-dimensional knots)& \pageref{page:subdivision}\\
\hline
$\subdivset = (s^{(1)}, \cdots, s^{(D)})$ & Subdivisions & \pageref{page:subdivisions} \\
\hline
$\setLSBj$ & Set of multi-indices for a block $\blockset_j$ & \pageref{page:multiset_j}\\
\hline
$\ul_j, \uk_j $ & Elements of $\setLSBj$  &  \pageref{page:ulj} \\
\hline
$\setTSBj$ & Set of knots in the multidimensional space $\R^{|\blockset_j|}$ from $S$ &  \pageref{page:hatfun} \\
\hline 
$\basisSi$ & Basis created from a subdivision $s^{(i)}$ & \pageref{page:basisSi}\\
\hline
$\continuous(X^{\blockset_j},\R)$ & Set of continuous functions depending on variables indexed in $\blockset_j$& \pageref{page:continuous:blocksetj}\\
\hline
$\hatfun$ & One-dimensional hat basis function with knots $u<v<w$ &  \pageref{page:hatfun} \\
\hline
$\hatfunSik$ & One-dimensional hat basis function element of $\basisSi$ & \pageref{page:hatfunSik}\\
\hline
$\hatfunlj$ & Multi-dimensional hat basis function & \pageref{page:hatfunlj}\\
\hline
$\spaceESP$ & Space spanned by the functions $(\phi_{\ul_j})_{\ul_j\in \setLSBj, 1\leq j \leq \sizeblock}$ &
\pageref{page:spaceESP}\\
\hline 
$\projSP$ & Projection onto the space $\spaceESP$ &  \pageref{page:projSP}\\
\hline 
$\GPj$ & GP defined on the space $\R^{|\blockset_j|}$, depending on the variables in $\blockset_j$ &  \pageref{page:GPj}\\
\hline
$\kernlj$ & Kernel associated to $\GPj$ &  \pageref{page:kernlj}\\
\hline
$\BAGP = \GP_1+ \cdots + \GP_\sizeblock$ & Block-additive GP &  \pageref{page:BAGP}\\
\hline
$\kernlBAGP$ & Kernel associated to $\BAGP$ &  \pageref{page:kernlBAGP}\\
\hline
$\GPSP$ & Projection of $\BAGP$ onto the space $\spaceESP$ &  \pageref{page:GPSP}\\
\hline
$\kernlSP$ & Kernel associated to $\GPSP$ &  \pageref{page:kernlSP}\\
\hline
$\predictor$ & MAP of $\GPSP$ conditioned to observations and constraints &  \pageref{page:predictor}\\
\hline
\end{tabular}
 \label{table:list:of:symbols}
\end{table}

\bibliography{mabib}
\bibliographystyle{abbrv}

\appendix

\section{Proof of Proposition~\ref{prop:algebraic expression}}\label{appendix:squareNormCriteria}

To start this section, we introduce two notations ( see \eqref{eq:basisSBj} and \eqref{eq:basisSP}) aiming to improve the readability of the proof. Recall that a pair of blocks and subdivisions $(\subpartition, \subdivset)$ define bases for $\blockset_j \in \subpartition=\{\blockset_1, \cdots, \blockset_\sizeblock\},$
\begin{equation}
\label{eq:basisSBj}
\basisSBj:=(\phi_{\ul_j})_{\ul_j \in \setLSBj},
\end{equation}
and a general basis 
\begin{equation}
\label{eq:basisSP}
\basisSP:= \bigcup_{j=1}^{\sizeblock} \basisSBj.
\end{equation}

The aim is to compute the explicit expression of the quantity in \eqref{eq:L2criterion}:
\[
\left\|\widehat{Y}^{\subdivset^\star}_{\subpartition^\star} - \widehat{Y}^\subdivset_{\subpartition}\right\|^2_{L^2}.
\]

For the sake of readability, in this proof, we simplify the notations by removing the indexes ``$\subdivset$'' and ``$\subpartition$'' on every object. Variables denoted with a superscript $\star$ refer to the couple $(\subpartition^\star, \subdivset^\star)$ as defined in Section \ref{subsec:choicesMaxMod}. Variables without the superscript refer to the couple $(\subpartition, \subdivset)$. This leads to the following notations:

\begin{itemize}
\item $\widehat{Y} = \widehat{Y}^\subdivset_{\subpartition}$ 
%= \widehat{\bxi}^\top \bPhi(\bX)$
(similarly, $\widehat{Y}^\star = \widehat{Y}^{\subdivset^\star}_{\subpartition^\star}$),
%= \widehat{\bxi}^{\star\top} \bPhi^\star(\bX)$
\item $\setL_j = \setLSBj$ (similarly, $\setL^\star_j = \setLSBjstar$),
\item $\widehat{Y}_j = \sum_{\ul_j \in \setL^{\subdivset}_{\blockset_j}} \bxi_{\ul_j} \phi_{\ul_j}$ (similarly, $\widehat{Y}^\star_j = \sum_{\ul_j \in \setL^{\subdivset^\star}_{\blockset_j^\star}}\bxi^\star_{\ul_j} \phi^\star_{\ul_j}$) where $\bxi:=(\bxi_1,\ldots, \bxi_\sizeblock)$ and $\bxi:=(\bxi^\star_1,\ldots, \bxi^\star_\sizeblock)$ are solution of \eqref{eq:xi:mod:sol},
\item $\setL = \setL^{\subdivset}_{\subpartition} = \bigcup_{j=1}^{\sizeblock} \setL_j$ (similarly, $\setL^\star
=
\setL^{\subdivset^\star}_{\subpartition^\star}
= \bigcup_{j=1}^{\sizeblock^\star} \setL^\star_j$),
\item $\bPhi_j := (\phi_{\ul_j})_{\ul_j \in \setLSBj}$ (similarly, $\bPhi^\star_j =(\phi^\star_{\ul_j})_{\ul_j \in \setLSBjstar}$) where $\phi_{\ul_j}$ and $\phi^\star_{\ul_j}$ are defined by \eqref{eq:phi:SBj}.
\end{itemize}

\subsection{Change of basis when updating the subdivision and/or the partition}

As a preliminary result,
we study the  change of basis of hat functions associated to two different pairs of blocks and subdivisions $(\subpartition, \subdivset)$ and $(\subpartition^\star,\subdivset^\star)$. The latter pair can be obtained after one iteration of the MaxMod algorithm from $(\subpartition, \subdivset)$ defined in Section~\ref{subsec:choicesMaxMod}. 
These changes of basis functions extend to several blocks of variables the results presented in \cite[Section SM2]{Bachoc2022MaxMod}.
Roughly speaking, an important idea is that a one-dimensional piecewise affine function $f$ defined on a subdivision remains  piecewise affine when defined on a finer subdivision. Furthermore, to express $f$ with the hat basis functions of the finer subdivision, it is sufficient to consider the values of $f$ on its knots.

\begin{lemma}[Expression of the elements of  $\basis^{\subdivset}_{\subpartition}$ in $\basis^{\subdivset^\star}_{\subpartition^\star}$]\label{lem:changebasis}
For $\subdivset=(s^{(1)},\ldots, \, s^{(D)})$ and $\subpartition=\{\blockset_1,\ldots, \blockset_\sizeblock\}$, we have the following explicit expressions of the basis functions in $\basisSP$ in the new basis $\basisstar$ for every choice of MaxMod introduced in Section \ref{subsec:choicesMaxMod}:  

\begin{description}
\item[Activate] 
Let $i_0$ be the index of the activated variable. Recall that this variable forms a new block.
Thus $\subpartition^\star=\{\blockset_1, \cdots, \blockset_\sizeblock, \{i_0\}\}$ and
$\basisstar$ is the set of functions 
\[ 
\basisstar = \basisSP \cup  \basis^{\subdivset^\star}_{\{i_0\}},
\]
where $\basis^{\subdivset^\star}_{\{i_0\}}=\{\bx \mapsto \widehat{\phi}_{-1,0,1}(x_{i_0}),\bx \mapsto \widehat{\phi}_{0,1,2}(x_{i_0})\}$. Hence, every function in $\basisSP$ lies in $\basisstar$.

\item[Refine] Let $s^{(i_0)}$ be the refined subdivision in the block $ \blockset_{j_0}$. Write $p^\star$ for the index of the left-nearest neighbor knot to $t^\star$ in the subdivision $s^{(i_0)}$: $t^{(i_0)}_{p^\star} <t^\star<t^{(i_0)}_{p^\star+1}$.
For any $j \neq j_0$, the elements $\phi_{\ul_j}$ of $\basisSBj$ are already in $\basisstar$. Consider a multi-index $\ul_{j_0}=(\ell_1,...,\ell_{|\blockset_{j_0}|}) \in \setL^{\subdivset}_{\blockset_{j_0}}$.  Without loss of generality we assume that the variable $i_0$ is the first in the block $\blockset_{j_0}$, with corresponding knots indexed by $\ell_1$ in $\ul_{j_0}$.
Let $\delta_{i_0} = (1,0,\ldots,0) \in \mathbb{R}^{|\blockset_{j_0}|}$.
If $\ell_{1} \not \in \{p^\star, p^\star+1\}$ then $\phi_{\ul_{j_0}} \in \basisstar$.
If $\ell_1=p^\star$ then
\[
\phi_{p^\star} = \phi^{\star}_{p^\star} + \phi_{p^\star}(t^\star) \phi^{\star}_{p^\star+1},
\]
which can be checked by computing the values at the knots of the finest subdivision $s^{\star(i_0)}$ (of indices $p^\star$ and $p^\star + 1$).
Similarly if $\ell_1=p^\star+1$, then
$\phi_{p^\star+1} =  \phi_{p^\star+1}(t^\star) \phi^{\star}_{p^\star+1} 
+
\phi^{\star}_{p^\star+2}$.
Finally for these two latter cases, we deduce, by tensorization,
\[
\phi_{\ul_{j_0}}
=
\begin{cases}
\phi^\star_{\ul_{j_0}} + \phi_{p^\star}(t^\star)
\phi^\star_{\ul_{j_0}+\delta_{i_0}}
& ~ \mbox{if} ~ \ell_1 = p^\star
\\ 
\phi_{p^\star+1}(t^\star)
\phi^\star_{\ul_{j_0}} + \phi^\star_{\ul_{j_0}+\delta_{i_0}}
& ~ \mbox{if} ~ \ell_1 = p^\star+1.
\end{cases}
\]
%Where $\ul_{j_0}+\delta_{i_0}$ is the vector $\ul_{j_0}$ on which we add $1$ to the $i_0$-th composant:
%\[
%\ul_{j_0}+\delta_{i_0}= (\ell_1, \ldots , \ell_{i_0-1}, \ell_{i_0}+1, \ell_{i_0+1}, \ldots, \ell_{|\blockset_{j_0}|}).
%\]

\item[Merge] In the case where we merge two blocks, suppose without loss of generality that $\blockset_1$ and $\blockset_2$ are merged, so that $\blockset^\star_1=\blockset_1 \cup \blockset_2$. For $j=1,2$, let $\blockset_j=\{i_{j,1},i_{j,2}, \ldots, i_{j,{|\blockset_j|}}\}$
with $i_{j,1} < \cdots < i_{j,{|\blockset_j|}}$. 
Then suppose that the elements in $\blockset^\star_1$ are ordered as $\blockset^\star_1 =\{i_{1,1}, \ldots, i_{1,{|\blockset_1|}}, i_{2,1}, \ldots, i_{2,{|\blockset_2|}}\}$.   
For any $j>2$, the basis functions $\phi_{\ul_j}\in \basisSBj$ are in $\basisstar$, since the block $\blockset_j$ is not modified by the merge. Now, consider $\ul_1=(\ell_1,...,\ell_{|\blockset_1|})\in \setL^\subdivset_{\blockset_1}$. 
Using that the hat basis functions corresponding to a block sum to one, the following equality holds:
\begin{eqnarray*}
\phi_{\ul_1}(\bx)
& = &
\left(
\prod_{a = 1}^{|\blockset_1|}
\phi_{\ell_a}^{(s^{(i_{1,a})})}(x_{i_{1,a}})
\right)
\cdot 
1
\\
& = &
\left(
\prod_{a = 1}^{|\blockset_1|}
\phi_{\ell_a}^{(s^{(i_{1,a})})}(x_{i_{1,a}})
\right)
\cdot 
\left(
\sum_{\ul_2 \in \setL^\subdivset_{\blockset_2} } 
\phi_{\ul_2}^{(\blockset_2)}
(x_{\ul_2})
\right)
 = 
\sum_{\ul^\star \in \setL^\star_{\ul_1}} \phi^\star_{\ul^\star}(\bx), 
\end{eqnarray*}
with $\setL^\star_{\ul_1}:= \{(\ell^\star_1,\ldots, \ell^\star_{|\blockset_1|+|\blockset_2|}) \in \setL^\subdivset_{\subpartition^\star}, (\ell^\star_1,\ldots, \ell^\star_{|\blockset_1|})=\ul_1\}$. Similarly, for $\phi_{\ul_2} \in \basis^\subdivset_{\blockset_2}$,
\[
\phi_{\ul_2} =\sum_{\ul^\star \in \setL^\star_{\ul_2}} \phi^\star_{\ul^\star},
\]
where $\setL^\star_{\ul_2}:=\{(\ell^\star_1,\ldots, \ell^\star_{|\blockset_1|+|\blockset_2|}) \in \setL^\subdivset_{\subpartition^\star}, (\ell^\star_{|\blockset_1|+1},\ldots, \ell^\star_{|\blockset_1|+|\blockset_2|})=\ul_2\}$.
\end{description} 
\end{lemma}

%\begin{proof}
%The proof is constructive, to get a deeper understanding of this construction we encourage curious reader to refer to the proof of \ref{prop:base inclusion}.
%\end{proof}

\begin{corollary}\label{corol:xi'}
From Lemma \ref{lem:changebasis}, 
a linear combination of the former basis functions from $(\subdivset,\subpartition)$ is also a linear combination of the new basis functions from $(\subdivset^\star,\subpartition^\star)$. Formally,
for every vector $\widehat{\bxi} \in \R^{|\setL|}$ there exists a vector $\widehat{\bxi}'$ in $\R^{|\setL^\star|}$, obtained by the change of basis formula, such that 
\[
\bPhi^\top\widehat{\bxi}=\bPhi^{\star\top}\widehat{\bxi}',
\]
where $\bPhi$ (respectively $\bPhi^\star$) are the vector functions introduced in~\eqref{eq:bPhiSP} for the subpartition $\subpartition$ (respectively $\subpartition^\star$) and subdivision $\subdivset$ (respectively $\subdivset^\star$).
\end{corollary}

\subsection{Computation of the L2Mod criterion}

Since our model is block additive, we have $\big(\widehat{Y}^{\star}-\widehat{Y}\big)^2=\left(\sum_{j=1}^{\sizeblock} \left[ \widehat{Y}_{j}-\widehat{Y}^\star_{j} \right] \right)^2$. By expanding the square and integrating, we deduce:
\begin{equation}
\label{eq:square norm criteria developped}
\left\|\widehat{Y}^{\star}-\widehat{Y}\right\|^2_{L^2}= \underbrace{\sum_{j=1}^{\sizeblock}\int_{[0,1]^{\blockset^\star_j}}\big(\widehat{Y}_{j}-\widehat{Y}^\star_{j}\big)^2\,d \lambda}_{S_1}+ 2\underbrace{\sum_{1\leq i<j\leq \sizeblock} \Big(\int_{[0,1]^{\blockset^\star_i}}(\widehat{Y}_{i}-\widehat{Y}^\star_{i})\, d \lambda \Big) \Big(\int_{[0,1]^{\blockset^\star_j}}(\widehat{Y}_{j}-\widehat{Y}^\star_{j})\, 
d \lambda \Big)}_{S_2},
\end{equation}
where $d \lambda$ is Lebesgue measure in the appropriate dimension. In $S_2$, we have exploited that the blocks are disjoint to write integrals of products as products of integrals.
We now investigate both sums $S_1$ and $S_2$ of \eqref{eq:square norm criteria developped} separately. Our approach for computing these two sums is to express $\widehat{Y}_{i}$ and $\widehat{Y}^\star_{i}$ in the ``finest'' basis corresponding to $\widehat{Y}^\star_{i}$ and 
to use the change of basis formulas of Lemma
\ref{lem:changebasis}.

\subsubsection{Computation of $\mathbf{S_1}$} \label{sssec:firstsum}

From Corollary \ref{corol:xi'}, we can consider the vector $\widehat{\bxi}'$, which satisfies $\bPhi^\top\widehat{\bxi}=\bPhi^{\star\top}\widehat{\bxi}'$.
Note that from \eqref{eq:predictor},
$\widehat{Y}(\bx) = \bPhi^{\top}(\bx)\widehat{\bxi}$
and
$\widehat{Y}^{\star}(\bx) = \bPhi^{\star\top}(\bx)\widehat{\bxi}'$. 
We then rewrite the sum $\mathbf{S_1}$ of \eqref{eq:square norm criteria developped} as
\begin{eqnarray*}
\sum_{j=1}^\sizeblock
\int_{[0,1]^{\blockset^\star_j}}(\widehat{Y}_{j}-\widehat{Y}^\star_{j})^2 \, d\lambda &=&
\sum_{j=1}^\sizeblock \int_{[0,1]^{\blockset^\star_j}} 
\Big(\sum_{\ul_j \in \setL^\star_j} (\widehat{\bxi}'_{j,\ul_j}-\widehat{\bxi}^\star_{j,\ul_j})\phi^{\star}_{\ul_j}\Big)^2 d\lambda\\
&=& 
\sum_{j=1}^\sizeblock \sum_{\ul_j,\ul'_j\in\setL^\star_j}
(\widehat{\bxi}'_{j,\ul_j}-\widehat{\bxi}^\star_{j,\ul_j})
(\widehat{\bxi}'_{j,\ul'_j}-\widehat{\bxi}^\star_{j,\ul'_j})
\int_{[0,1]^{\blockset^\star_j}}
\phi^{\star}_{\ul_j}\phi^{\star}_{\ul'_j}\, d \lambda.
\end{eqnarray*}
Now, we define the $|\setL^\star|$-dimensional vector $\beeta$ as
\begin{equation}\label{eq:eta vector}
\beeta=(\beeta_{j,\ul_j})_{1\leq j\leq B, \ul_j \in \setL^\star_j}, \hspace{1cm} \beeta_{j,\ul_j} = (\widehat{\bxi}'_{j,\ul_j}-\widehat{\bxi}^\star_{j,\ul_j}),
\end{equation} 
and for $j=1, \cdots, \sizeblock$, the $|\setL^\star_j|$-dimensional matrix $\bPsi^{j}$ as
\begin{equation}\label{eq: bPsi block}
\bPsi^{j}_{\ul_j,\ul'_j} =
\int_{[0,1]^{\blockset^\star_j}}\phi^{\star}_{\ul_j}\phi^{\star}_{\ul'_j}d\lambda =
\prod_{i \in \blockset^\star_j}\int_0^1 \phi^{\star(i)}_{{\ul_{j,i}}}\phi^{\star(i)}_{{\ul'_{j,i}}}\, d\lambda =
\prod_{i \in \blockset^\star_j} \Psi^{(i)}_{\ul_{j,i},\ul'_{j,i}},
\end{equation}
with
\begin{equation}\label{eq: bPsi var}
    \Psi^{(i)}_{\ul_{j,i},\ul'_{j,i}} =
    \begin{cases}
        \dfrac{t^{(i)}_{\ul_{j,i}+1}-t^{(i)}_{\ul_{j,i}}}{3}     		& \text{if } \ul_{j,i}=\ul'_{j,i} = 1,  \\
        %\\
        \dfrac{t^{(i)}_{\ul_{j,i}+1}-t^{(i)}_{\ul_{j,i}-1}}{3}    		& \text{if } 2\leq\ul_{j,i}=\ul'_{j,i}\leq m^{(i)}-1, \\
        %\\
        \dfrac{t^{(i)}_{\ul_{j,k}}-t^{(i)}_{\ul_{j,i}-1}}{3}      		& \text{if } \ul_{j,i}=\ul'_{j,i} = m^{(i)},  \\
        %\\
        \dfrac{\big|t^{(i)}_{\ul_{j,i}}-t^{(i)}_{\ul'_{j,i}}\big|}{6}   & \text{if } |\ul_{j,i}-\ul'_{j,i}|=1,  \\
        %\\
        0 															 	& \text{if } |\ul_{j,i}-\ul'_{j,i}|>1.				
    \end{cases} 
\end{equation}
The expressions in \eqref{eq: bPsi var} correspond to the Gram matrices of univariate hat basis functions and can be found for instance in \cite{Bachoc2022MaxMod}. 
We finally get the result,

\begin{equation}\label{eq:sum1}
S_1 = 
\sum_{j=1}^B
\int_{[0,1]^{\blockset_j}}(\widehat{Y}_{j}-\widehat{Y}^\star_{j})^2 \, d\lambda = 
\sum_{j=1}^B
\sum_{\ul_j, \ul'_j \in \setL^\star_j} \beeta_{j,\ul_j}\beeta_{j,\ul'_j}\bPsi^{j}_{\ul_j,\ul'_j} =
\beeta^\top \bPsi \beeta,
\end{equation}
writing $\bPsi$ as the $|\setL^\star|$-dimensional matrix and $\beeta$ as the $|\setL^\star|$-dimensional vector
\begin{equation}\label{eq:Psi matrix}
\bPsi=
\begin{bNiceArray}{ccc}[margin] 
\bPsi^{1} &\Block{2-2}{0}  \\
\Block{2-2}{0} & \ddots \\
&& \bPsi^\sizeblock
\end{bNiceArray}, \hspace{1cm}
\beeta=
\begin{bNiceArray}{c}[margin] 
\bxi'_1 - \bxi_1 \\
\vdots\\
\bxi'_\sizeblock - \bxi_\sizeblock
\end{bNiceArray}.
\end{equation}

\begin{remark}
The computational cost of $S_1$ in \eqref{eq:sum1} is linear with respect to the dimension $|\setL^\star|$. 
Indeed,
for each $\ul_j \in \setL^\star_j$, there are at most $3^{|\blockset^\star_j|}$ multi-indices $\ul'_j \in \setL^\star_j$ such that $\bPsi^{j}_{\ul,\ul'}\neq 0$.
%\end{remark}
%\begin{proof}
This is because, from Equations~\eqref{eq: bPsi block} and \eqref{eq: bPsi var}, for $\ul_j \in \setL^\star_j$ , it is easy to see that $\bPsi^{j}_{\ul_j,\ul'_j}\neq 0$ implies that $||\ul_j - \ul_j'||_{\infty}\leq 1$. 
Since $\ul_j$ and $\ul'_j$ both lie in $\Z^{|\blockset^\star_j|}$, the number of values that $\ul'_j$ can take such that $\bPsi^{j}_{\ul_j,\ul'_j}\neq 0$  is bounded by $3^{|\blockset^\star_j|}$.
Hence, the number of non-zero terms in $\displaystyle{\sum_{\ul_j, \ul'_j \in \setL^\star_j} \beeta_{j,\ul_j}\beeta_{j,\ul'_j}\bPsi^{j}_{\ul_j,\ul'_j}}$ in \eqref{eq:sum1} is bounded by $|\setL^\star_j|3^{|\blockset^\star_j|}$. 
%From the block-diagonal structure of $\bPsi$, the number of non zero elements in the second sum in Equation~\eqref{eq:sum1} is bounded by $|\setL^\star_j| 3^{|\blockset^\star_j|}$.
%\end{proof}
\end{remark}

\subsubsection{Computation of $\mathbf{S_2}$}\label{sec:secondsum}

We define the $|\setL^\star|$-dimensional vector $\bE$,
\begin{equation}\label{eq:Evect}
\bE\in \R^{|\setL^\star|}, \hspace{1cm} \bE=(\bE_{j,\ul_j})_{1\leq j\leq B, \,\ul_j \in \setL^\star_j},
\end{equation}
with
\begin{equation}\label{eq:E_ul}
\bE_{\ul_j}=\int_{[0,1]^{\blockset^\star_j}} \phi^{\star}_{\ul_j} \, d\lambda = 
\int_{[0,1]^{\blockset^\star_j}} \prod_{i \in \blockset_j}\phi^{s^{\star(i)}}_{\ul_{j,i}}\, d\lambda =
\prod_{i \in \blockset^\star_j} \bE_{j,\ul_{j,i}},
\end{equation}
and $\bE_{j,\ul_{j,i}}=\int_{0}^1\phi^{s^{\star(i)}}_{\ul_{j,i}}\, d \lambda$. Then we can easily compute (see for instance \cite{Bachoc2022MaxMod})

\begin{equation}
    \bE_{j,\ul_{j,i}} = 
    \begin{cases}
        \frac{1}{2} (t^{(i)}_{\ul_{j,i}+1}-t^{(i)}_{\ul_{j,i}})  & \text{if }   \ul_{j,i} = 1, \\
        \frac{1}{2} (t^{(i)}_{\ul_{j,i}+1}-t^{(i)}_{\ul_{j,i}-1})& \text{if }   2\leq \ul_{j,i}\leq m^{\star(i)}-1, \\
        \frac{1}{2} (t^{(i)}_{\ul_{j,i}}-t^{(i)}_{\ul_{j,i}-1})  & \text{if }   \ul_{j,i} = m^{\star(i)}.
    \end{cases}
\end{equation}
Then, taking $\widehat{\bxi}'$ from Corollary~\ref{corol:xi'} we can write, 
\[
\rho_j :=
\int_{[0,1]^{\blockset^\star_j}}(\widehat{Y}_{j}-\widehat{Y}^\star_{j})\, d\lambda=
\int_{[0,1]^{\blockset^\star_j}}\sum_{\ul_j \in \setL^\star_j}\phi^\star_{\ul_j}\big(\widehat{\bxi}'_{j,\ul_j}-\widehat{\bxi}^\star_{j,\ul_j}\big) \, d\lambda = 
\sum_{\ul_j \in \setL^\star_j} \beeta_{j,\ul_j} \bE_{\ul_j}
=
\beeta_j^\top \bE_j.
\]
\\
This gives
\[
S_2
=
2\sum_{1\leq i<j\leq \sizeblock}
\rho_i \rho_j
=
\left( \sum_{i=1}^B \rho_i
\right)^2
- 
\sum_{i=1}^B \rho_i^2
=
(\beeta^\top \bE)^2- \sum_{1\leq j \leq B}  \left(\beeta_j^\top \bE_j\right)^2.
\]
This gives an expression of $S_2$ that has a linear computational cost with respect to $B$.
Gathering the expressions of $S_1$ and $S_2$, together with \eqref{eq:square norm criteria developped}    concludes the proof of Proposition \ref{prop:algebraic expression}.

%\begin{proposition}\label{complete proposistion}
%Let $\widehat{Y}^\star$ and $\widehat{Y}$ being two functions in the vector spaces $E^{\subdivset^\star}_{\subpartition^\star}$ and $E^{\subdivset}_{\subpartition}$, with $\widehat{\bxi}'$ defined in corollary \ref{corol:xi'} such that $\widehat{Y}=\bPhi^\top\widehat{\bxi} = \bPhi^{\star\top}\widehat{\bxi}'$, $\widehat{\bxi}^\star$ such that $\widehat{Y}^\star=\bPhi^{\star\top}\widehat{\bxi}^\star$,   and $\bPsi$, $\bE$, $\beeta$ the objects introduced in the two last sections we have
%\begin{equation}
%\left\|\widehat{Y}^{\star}-\widehat{Y}\right\|^2_{L^2}=\beeta^\top \bPsi \beeta + (\beeta^\top \bE)^2 - \sum_{1\leq i\leq \sizeblock}\sum_{\ul_j \in \setL_j} (\beeta_{j,\ul_j}\bE_{\ul_j})^2.
%\end{equation}
%\end{proposition}

\section{Covariance parameters estimation} \label{appendix:Kernel Hyperparameters selection}

Here we consider a fixed partition $\partition=\{\blockset_1,\ldots, \blockset_\sizeblock\}$ and fixed subdivisions $\subdivset=\{s^{(1)},\ldots, s^{(\sizeblock)}\}$. We keep notations $X^{(i)}, X^{\blockset_j} ,X$ introduced in Section \ref{subsubsec:MultiDimensionalBasis}. 
We consider a parametric family of covariance functions
for the block-additive model \eqref{eq:BAGP}, given by \eqref{eq:kernlBAGP}. Within each block, we choose to tensorize univariate covariance functions. Formally, we let 
$$
    \kernl_{\partition,\theta}(\boldsymbol{x}, \boldsymbol{x}') = \sum_{j =1}^B
    \prod_{i\in \blockset_j} \kernl_{\theta^{(i)}}(x_i,x'_i),
$$
for $\boldsymbol{x},\boldsymbol{x}' \in X$ and $\theta=(\theta^{(1)},\ldots,\theta^{(D)})$, where for each $i$, $\kernl_{\theta^{(i)}}$ is a covariance function on $X^{(i)} \times X^{(i)}$ and $\Theta^{(i)} \subseteq \R^{q_i}$ for some $q_i \in \mathbb{N}$.

Then, we consider standard maximum likelihood estimation for the finite-dimensional GP $\GPSP$ in \eqref{eq:GPSP} with noisy observations, see \cite{LopezLopera2017FiniteGPlinear}. 
%
%Let us first write formally the parametric sets of covariance functions %$\mathcal{F}^{(i)}=\{k_{\theta^{(i)}}:X^{(i)}\times X^{(i)}\to \R,\, \theta^{(i)} \in \Theta^{(i)}\}$, for each coordinate $1\leq i \leq D$, where  For every $j$, we can naturally construct the covariance functions on $X^{\blockset_j}$ as
%\[
%\mathcal{F}^{\blockset_j}:=
%\bigotimes_{i\in \blockset_j}\mathcal{F}^{(i)}=
%\left\{\prod_{i\in \blockset_j} k_{\theta^{(i)}}\circ \Pi_i, \, \btheta^{\blockset_j} \in \Theta^{\blockset_j}\right\}
%\]
%\[
%(x_i)_{i \in \blockset_j}
%,
%(x'_i)_{i \in \blockset_j}
%\mapsto
%\prod_{i\in \blockset_j} k_{\theta^{(i)}}(x_i,x'_i)
%\]
%for
%$\btheta^{\blockset_j} \in \Theta^{\blockset_j}$
%with $\btheta^{\blockset_j} = (\theta^{(i)})_{i\in \blockset_j}$ and $\Theta^{\blockset_j}= \prod_{i\in \blockset_j} \Theta^{(i)}$. 
%Recall that the $\Pi_i$ are the natural surjections $X\to X^{(i)}$.  
%We then let as in \eqref{eq:kernelBlockAdd},
 %Note that $\theta \in \Theta$ with $\Theta= \Theta^{\blockset_1} \times \cdots \times \Theta^{\blockset_B}$.
%
%To select the global covariance parameter $\theta$, we will use the maximum likelihood method, 
Formally  we consider the finite-dimensional covariance function in \eqref{eq:kernlGPSP} that yields the finite-dimensional covariance matrix $\MatKtheta=\bPhi(\bX)^\top \widetilde{k}_{\theta}(\bX,\bX)\bPhi(\bX)$ of $\GPSP(\bX)$. The noisy observation vector $\bY =\GPSP(\bX)+\bepsilon$ is Gaussian $\mathcal{N}(
\boldsymbol{0} , \MatKtheta+\tau^2 \bI_n)$ and the associated likelihood is given by
\begin{equation}\label{eq:maxikelihood}
L(\theta, \tau ; \bY)  = 
\frac{1}{(2\pi)^{n/2}|\MatKtheta+\tau^2 \bI_n|^{1/2}} \exp\left(-\frac{1}{2}\bY^{\top}(\MatKtheta+\tau^2 \bI_n)^{-1}\bY\right).
\end{equation}
Numerical improvements can be used for computing the inverse and determinant of $\MatKtheta +\tau^2 \bI_n$, using the techniques of Section~\ref{subsec:Interpolation constraints and computation cost}.
Maximizing the likelihood over $(\theta,\tau)$ is equivalent to solving the optimization problem:
\[
\underset{
\substack{
\theta \in \Theta \\
\tau \in (0,\infty)
}}{\min}
~ ~
\log(|\MatKtheta+\tau^2 \bI_n|)+\bY^{\top}(\MatKtheta+\tau^2 \bI_n)^{-1}\bY.
\]
%In this optimization problem we seek to find the derivative of the function $L:(\theta, \tau) \mapsto \frac{1}{2}\log(|\MatKtheta+\tau^2\bI_n|)+\frac{1}{2}\bY^{\top}(\MatKtheta+\tau^2\bI_n)^{-1}\bY$,
To simplify its numerical resolution, we can provide the gradient which is given explicitly, see for instance \cite{Rasmussen2005GP} [Chap 5.4]:
\begin{equation}\label{eq:hyperparam theta diff}
\frac{\partial L(\theta, \tau ; \bY)}{\partial \theta_{j,\ell}} = 
- \bY^\top (\MatKtheta+\tau^2\bI_n)^{-1}\frac{\partial \MatKtheta}{\partial \theta_{j,\ell}}(\MatKtheta+\tau^2\bI_n)^{-1}\bY + \trace\left((\MatKtheta+\tau^2\bI_n)^{-1}\frac{\partial \MatKtheta}{\partial \theta_{j,\ell}}\right),
\end{equation}
and
\begin{equation}\label{hyperparam tau diff}
\frac{\partial L(\theta, \tau ; \bY)}{\partial \tau^2} = 
 \bY^\top (\MatKtheta+\tau^2\bI_n)^{-2}\bY +  \trace\left((\MatKtheta+\tau^2\bI_n)^{-1}\right),
\end{equation}
where for $j=1,\ldots,D$, $\theta_{j,1},\ldots,\theta_{j,q_j}$ are the components of $\theta_j$.
%Here the operator $\trace$ is the trace operator. 
%Having derivatives of $L$, we can find a minimizer using a gradient descent method.

%\[
%\left|\bPhi(\bX)^\top \widetilde{k}_{\theta}(\bX,\bX)\bPhi(\bX) + \tau^2 \bI_n \right| =
%\tau^{2n}\left|\widetilde{k}_{\theta}(\bX,\bX)\right|\left|\widetilde{k}_{\theta}(\bX,\bX)^{-1} +\tau^{-2}\bPhi(\bX)\bPhi(\bX)^\top\right|.
%\]
%The matrix $\widetilde{k}_{\theta}(\bX,\bX)$ is itself block diagonal, hence the computation of the determinant can be done faster when $n \ll |\setL^{\partition}_{\subdivset}|$. 

\section{Evolution of the square norm over iterations}\label{appendix:SEevolution}

Figure~\ref{fig:toolCriteria} suggests that the MSE score can occasionally increase over some of the MaxMod iterations. Here we show that this behavior is not caused by numerical issues, by providing theoretical examples where it occurs. We provide these theoretical examples in the unconstrained case, for simplicity, relying on the explicit expression of the mode in this case, which coincides with the usual conditional mean of GPs.

Let $X = [0,1]$.
Consider two finite vector subspaces $E_1$ and $E_2$ of $\continuous(X,\R)$ the realisation space of our GP $\{Y(x), \, x \in X\}$, satisfying $E_1 \subset E_2$. Suppose as well that $Y$ is a zero-mean GP with kernel $k$. We have two projections $P_1: \continuous(X,\R) \to E_1$ and $P_2: \continuous(X,\R) \to E_2$ such that $P_1\circ P_2= P_1$. We now set the two GPs $\widetilde{Y}_1 = P_1(Y)$ and $\widetilde{Y}_2 = P_2(Y)$. We set our observations to be $(\bX,\bY)$. One may think that the function $\widehat{f}_2(\cdot)=E(\widetilde{Y}_2(\cdot)|\widetilde{Y}_2(\bX)+\epsilon=\bY)$ better interpolates the observations $\bY$ than the function $\widehat{f}_1(\cdot)=E(\widetilde{Y}_1(\cdot)|\widetilde{Y}_1(\bX)+\epsilon=\bY)$ ($\epsilon$ being a Gaussian white noise of variance $\tau^2$). Indeed, $\widetilde{Y}_2$ lives in a larger vector space than $\widetilde{Y}_1$. However, as already mentioned, Figure~\ref{fig:toolCriteria} shows some occasional increments of the square norm
from $\widetilde{Y}_1$ to $\widetilde{Y}_2$.
To interpret these increments, it is convenient to recall that for $i=1,2$ the conditional mean function $\widehat{f}_i$ can also be defined as the solution of a minimization problem in the RKHS $\m{H}_i$ with kernel $\widetilde{k}_i$:
\begin{equation}\label{eq:minimization problem}
\widehat{f}_i= 
\underset{f \in \m{H}_i}{\mathrm{argmin}} \left\|f(\bX)-\bY\right\|^2 +\tau^2\left\|f\right\|^2_{\m{H}_i}.
\end{equation}
%Here $\m{H}_i$ is the Hilbert space spanned by $(\widetilde{k}_i(x,\cdot))_{x\in X}$,
Here $\widetilde{k}_i$ is the (finite-dimensional) kernel of $\widetilde{Y}_i$. Notice  that $\m{H}_1 \subseteq \m{H}_2$ since $E_1 \subseteq E_2$. Hence, we have 
\[
\underset{f \in \m{H}_2}{\min} \left\|f(\bX)-\bY\right\|^2
\le 
\underset{f \in \m{H}_1}{\min} \left\|f(\bX)-\bY\right\|^2.
\]
However, due to the second term in \eqref{eq:minimization problem}, we can have 
\begin{equation}\label{eq:inequality counter example}
||\widehat{f}_1(\bX)-\bY||^2 < ||\widehat{f}_2(\bX)-\bY||^2.
\end{equation}
We now give an explicit example where this happens. We will find two hat basis $\basis_1$ and $\basis_2$ such that $E_1=\operatorname{span} \basis_1 \subset E_2 = \operatorname{span} \basis_2$ and such that \eqref{eq:inequality counter example} holds.

Note that Section \ref{subsec:Interpolation constraints and computation cost} provides an explicit expression of the function $\widehat{f}_i$ and thus we have
\[
\left\|\widehat{f}_i(\bX)-\bY\right\|^2= \left\| \left(\widetilde{k}_i(\bX,\bX) \Big[\widetilde{k}_i(\bX,\bX) + \tau^2 \bI_n\Big]^{-1} - I_n \right)\bY\right\|^2.
\]
To obtain that construction, we first show Lemma \ref{lem:IneqCond} that will be useful in the following developments.

\begin{lemma}\label{lem:IneqCond}
Let $\bA$ and $\bB$ be two symmetric $n \times n$ matrices. If the matrix $\bB-\bA$ has one strictly positive eigenvalue $\lambda$ with an associated unit eigenvector $\be_{\lambda}$, then:
\[
\left\| (\bA[\bA+\gamma \bI_n]^{-1} - \bI_n)\be_\lambda \right\| > \left\| (\bB[\bB+\gamma \bI_n]^{-1} - \bI_n)\be_\lambda\right\|
\]
holds when $\gamma$ is large enough.
\end{lemma}
\begin{proof}
We can rewrite $[\bA+\gamma \bI_n]^{-1}= \gamma^{-1} [\bI_n + \frac{\bA}{\gamma}]^{-1}$. Then, as $\gamma \to \infty$,
\[
[\bA+\gamma \bI_n]^{-1}= \gamma^{-1} \left(\bI_n -\frac{\bA}{\gamma} + o\left(\frac{ 1 }{\gamma}  \right)\right),
\]
and again 
\[
\bA[\bA+\gamma \bI_n]^{-1} - \bI_n
=
\frac{\bA}{\gamma} - \bI_n + o\left(\frac{ 1 }{\gamma}  \right).
\]
Note that the same expression holds for $\bB$. These expressions provide the following equalities:
\begin{eqnarray*}
\left\| (\bA[\bA+\gamma \bI_n]^{-1} - \bI_n)\be_\lambda \right\|^2 -
\left\| (\bB[\bB+\gamma \bI_n]^{-1} - \bI_n)\be_\lambda\right\|^2 
&=& \\
\left\| \left(\frac{\bA}{\gamma}-\bI_n + o\left(\frac{ 1 }{\gamma}  \right)\right)\be_\lambda\right\|^2 - \left\|\left(\frac{\bB}{\gamma}-\bI_n + o\left(\frac{ 1 }{\gamma}  \right)\right)\be_\lambda \right\|^2
&=&
\frac{2}{\gamma} \langle \be_\lambda , (\bB-\bA)\be_\lambda \rangle + o(\gamma^{-1})\\
&=&
\frac{2 \lambda}{\gamma}
+ o(\gamma^{-1}),
\end{eqnarray*}
concluding the proof.
\end{proof}

We do now have a way of constructing our inequality. Taking $\basis_1=(\phi_1)=(\widehat{\phi}_{0,0.5,1})$ and $\basis_2 = (\phi'_1, \phi'_2)= (\widehat{\phi}_{0,0.5,0.5+\epsilon}, \widehat{\phi}_{0.5,0.5+\epsilon,1})$, we have $\phi_1 = \phi'_1 + (1-2\epsilon)\phi'_2 $.
Indeed, since $\phi_1, \phi'_1,\phi'_2$ are piecewise linear vanishing at $0,1$ it is sufficient to check the equality at the knots $0.5$ and $0.5+\epsilon$. 
In particular, we can express $\phi_1$ in the basis $\basis_2$ and thus $E_1 \subseteq E_2$.
Then, taking $\bX=(x_1, x_2)= (0.5, 0.5+\epsilon)$ gives
\[
\Phi_2(\bX)^\top=
\begin{bNiceArray}{cc}[margin] 
\phi'_1(x_1) & \phi'_2(x_1) \\
\phi'_1(x_2) & \phi'_2(x_2) 
\end{bNiceArray} = I_2.
\]
From what we said 
\[
\Phi_1(\bX)^\top=
\begin{bNiceArray}{c}[margin] 
\phi_1(x_1)  \\
\phi_1(x_2)
\end{bNiceArray} =
\begin{bNiceArray}{c}[margin] 
\phi'_1(x_1) +  (1-2\epsilon)\phi'_2(x_1) \\
\phi'_1(x_2) +  (1-2\epsilon)\phi'_2(x_2) \\
\end{bNiceArray} = 
\begin{bNiceArray}{c}[margin] 
1 \\
1-2\epsilon \\
\end{bNiceArray}. 
\]
We can then express $\widetilde{k}_1(\bX,\bX)$:
\[
\widetilde{k}_1(\bX,\bX)=
\begin{bNiceArray}{c}[margin] 
1 \\
1-2\epsilon 
\end{bNiceArray}
\kernl (\bX,\bX)
\begin{bNiceArray}{cc}[margin] 
1 &  1-2\epsilon 
\end{bNiceArray}.
\]

Finally we want to show that the matrix $\widetilde{k}_1(\bX,\bX)-\widetilde{k}_2(\bX,\bX)$ has some strictly positive eigenvalues:
\[
\widetilde{k}_1(\bX,\bX)-\widetilde{k}_2(\bX,\bX)=
\begin{bNiceArray}{c}[margin] 
1 \\
1-2\epsilon 
\end{bNiceArray}
\kernl (\bX,\bX)
\begin{bNiceArray}{cc}[margin] 
1 &  1-2\epsilon 
\end{bNiceArray}
-\kernl (\bX,\bX).
\]
If $\kernl (\bX,\bX) = I_2$ and $\epsilon=0$ the matrix $\widetilde{k}_1(\bX,\bX)-\widetilde{k}_2(\bX,\bX)$ has for eigenvalues $\{-1,1\}$ thus there is one strictly positive eigenvalue. By continuity of the largest eigenvalue for symmetric matrices there exists $\epsilon>0$ and a kernel $k$ such that the above matrix has strictly positive eigenvalues. This constructs the counter example we were looking for by applying Lemma \ref{lem:IneqCond} with $A = \widetilde{k}_2(\bX,\bX)$, $B = \widetilde{k}_1(\bX,\bX)$, $\gamma = \tau^2$ and $\bY = \be_{\lambda}$.

\section{Change of basis: A generalisation}
\label{ann:ChangeBasisGenera}

We focus here on the generalization of the Lemma~\ref{lem:changebasis}. Two pairs of blocks and subdivsions $(\subpartition,\subdivset),(\subpartition^\star, \subdivset^\star)$ provide two bases $\basisSP, \basisSPstar$ defined in \eqref{eq:basisSP}. We provide necessary and sufficient conditions to be able to express any element in $\basisSP$ in $\basisstar$. In other words, we provide necessary and sufficient condition so that $E^\subdivset_\subpartition \subset E^{\subdivset^\star}_{\subpartition^\star}$.

\begin{lemma}[Change of basis from  $\basis^{\subdivset}_{\subpartition}$ to $\basis^{\subdivset^\star}_{\subpartition^\star}$ ]\label{prop:base inclusion}
Let $\subdivset=(s^{(1)},\ldots, \, s^{(D)})$ and $\subdivset^{\star}=(s^{\star(1)},\ldots, \, s^{\star(D)})$ be two subdivisions with associated subpartition $\subpartition=\{\blockset_1,\ldots, \blockset_\sizeblock\}$ and $\subpartition^\star=\{\blockset^{\star}_1,\ldots,\blockset^{\star}_{\sizeblock^{\star}}\}$, let $E^{\subdivset}_{\subpartition}$ and $E^{\subdivset^\star}_{\subpartition^\star}$ be the the vector spaces defined in \eqref{eq:spaceESP:setLSP}.\\
Then, $E^{\subdivset}_{\subpartition}\subset E^{\subdivset^\star}_{\subpartition^\star}$ if and 
only if the two following conditions are satisfied: 
\begin{itemize}
\item[(i)] \textbf{subdivision inclusion}: For each $i \in \bigcup_{j=1}^{\sizeblock} \blockset_j$, we have $s^{(i)}\subset s^{\star(i)}$. 
\item[(ii)] \textbf{subpartitions inclusion}: For each $\blockset_j$ block set in $\subpartition$, there exists $j^{\star}$ such $\blockset_j \subset \blockset^{\star}_{j^\star}$. 
\end{itemize} 
Thus there is an algorithm
giving the change of basis matrix $P_{\basis^\subdivset_{\subpartition},\basis^{\subdivset^\star}_{\subpartition^\star}}$ where bases $\basis^\subdivset_{\subpartition}$ and $\basis^{\subdivset^\star}_{\subpartition^\star}$ are defined in Section \ref{subsubsec:MultiDimensionalBasis}.  
\end{lemma}
\begin{proof}
\textbf{( Sufficient condition $\Longrightarrow$)}\\
We present an algorithmic proof by simplifying the problem in stages. First, consider the case where there is only one variable, and that $s^{(1)}\subset s^{\star(1)}$. For any basis function $\phi \in \basis^{(1)}$, we can express it as a linear combination of functions in the basis $\basis^{\star(1)}$ as follows:
\[
\phi =
\sum_{k=1}^{m^{\star(1)}}\phi(t^{\star(1)}_k)\phi^{\star(1)}_{k}.
\]
This representation is intuitive since it projects the linear-by-parts function $\phi$ from the basis $\basis^{(1)}$ onto the linear-by-parts functional space spanned by $\basis^{\star(1)}$ which is more ``precise''.

\paragraph{Case 1: Refinement.} When the subpartitions are identical, for $1\leq j\leq \sizeblock$, we can express $\phi_{\ul_j}\in E^{S}_{\subpartition}$ in the vector space $E^{\subdivset^\star}_{\subpartition^\star}$ as follows:
\[
\phi_{\ul_j}=\prod_{i \in \blockset_j}\left(\sum_{k=1}^{m^{\star(i)}} \phi^{(i)}_{\ul_{j,i}}(t^{\star(i)}_{\ul_{j,k}})\phi^{\star(i)}_{k}\circ \Pi_i\right),
\]
here $\Pi_i$ is the canonical surjection $X\to X^{(i)}$.
By expanding the product, it becomes clear that $\phi_{\ul_j}$ belongs to $E^{\subdivset^\star}_{\subpartition^\star}$, and we can define the matrix of change of basis as $P_{\basis^\subdivset_{\subpartition},\basis^{\subdivset^\star}_{\subpartition^\star}}$.

\paragraph{Case 2: Activating/Merging.} Consider the case where the subpartition $\subpartition$ consists of blocks $\blockset_j$ such that $\blockset_j \subset \blockset^{\star}_{j^\star}$, and for every $i \in \blockset_j$, we have $s^{(i)}=s^{\star(i)}$. We observe that:
\[
\setL^{\subdivset^\star}_{\blockset^\star_{j^\star}}= \setL^{\subdivset}_{\blockset_j} \times \setL^{\subdivset^\star}_{\blockset^\star_{j^\star}\setminus \blockset_j},
\]
which allows us to express any element $\ul^\star_j \in \setL^{\subdivset^\star}_{\blockset_j}$ as $\ul^\star_j=(\ul_a,\ul_b)$, where $(\ul_a,\ul_b)\in \setL^{\subdivset}_{\blockset_j} \times \setL^{\subdivset^\star}_{\blockset^\star_j\setminus \blockset_j}$. Noticing that for every $i\in \blockset^\star_{j^\star}$, $\sum_{k=1}^{m^{\star(i)}}\phi^{\star(i)}_k =1$, we can express any basis function $\phi_{\ul_j} \in \basisSBj$ as 
\[
\phi_{\ul_j}=\prod_{i \in \blockset_j}\phi^{(i)}_{\ul_{j,i}}\circ\Pi_i \prod_{i\in\blockset^\star_{j^\star}\setminus\blockset_j}\left(\sum_{k=1}^{m^{\star(i)}}\phi^{\star(i)}_k\circ\Pi_i\right),
\]
again for every $i=1, \cdots, D$, $\Pi_i$ is the canonical surjection $X \to X^{(i)}$. The last equality, upon expansion of the last, yields:
\[
\phi_{\ul_j}=\sum_{\ul_b \in \setL^{\subdivset^\star}_{\blockset^\star_{j^\star}\setminus \blockset_j}} \phi^\star_{(\ul_j,\ul_b)}.
\]

\noindent \textbf{General case:} We can now reconstruct the change of basis matrix by constructing intermediate bases. \textbf{case 1} provides us with $P_{\basis^{\subdivset}_{\subpartition},\basis^{\subdivset^\star}_{\subpartition}}$. We can then apply \textbf{case 2} to obtain the matrix $P_{\basis^{\subdivset^\star}_{\subpartition},\basis^{\subdivset^\star}_{\subpartition^\star}}$. Finally, we have:
\[
P_{\basis^\subdivset_{\subpartition},\basis^{\subdivset^\star}_{\subpartition^\star}} = P_{\basis^{\subdivset^\star}_{\subpartition},\basis^{\subdivset^\star}_{\subpartition^\star}}P_{\basis^{\subdivset^\star}_{\subpartition},\basis^{\subdivset}_{\subpartition}}.
\]

\noindent \textbf{(Necessary condition $\Longleftarrow$)}\\

On the other way, let us consider that conditions are not met and reach a contradiction. \\
\textbf{Non subdivision inclusion:} There is $i\in \disjcup_{j=1}^{\sizeblock}\blockset_j$ such that $s^{(i)}\not\subset s^{\star(i)}$ it means that there is $t^{(i)}_{k}$ in $s^{(i)}$ which is not in $s^{\star(i)}$. By remarks made in \textbf{case 1.} we have that the function $\phi^{(i)}_k: \R \to \R$, $x \mapsto \phi^{(i)}_\kernl (x)$ is in the space $E^{\subdivset}_{\subpartition}$. It is clear it is not in the space $E^{\subdivset^\star}_{\subpartition^\star}$. Otherwise, by projection property would give: 
\[
\phi^{(i)}_k =\sum_{l=1}^{m^{\star(i)}} \phi^{(i)}_\kernl (t^{\star(i)}_l)\phi^{\star(i)}_{l}.
\]
However, as $t^{(i)}_k \in [0,1]$, for some $1 \leq l^\star\leq m^{\star(i)}$ the following inequality holds: $t^{\star(i)}_{l^\star}<t^{(i)}_k<t^{\star(i)}_{l^\star+1}.$ Thus 
\[
\Bigg(\sum_{l=1}^{m^{\star(i)}} \phi^{(i)}_\kernl (t^{\star(i)}_l)\phi^{\star(i)}_{l} \Bigg)(t^{(i)}_k) = \phi^{(i)}_\kernl (t^{\star(i)}_{l^\star})\phi^{\star(i)}_{l}(t^{(i)}_k) + \phi^{(i)}_\kernl (t^{\star(i)}_l)\phi^{\star(i)}_{l^\star+1}(t^{(i)}_k)< 1 =\phi^{(i)}_\kernl (t^{(i)}_k) .\]
It is now clear by construction of $E^{\subdivset}_{\subpartition}$ and $E^{\subdivset^\star}_{\subpartition^\star}$ that we do not have $E^{\subdivset}_{\subpartition} \subset E^{\subdivset^\star}_{\subpartition^\star}$.\\ 
\textbf{Non subpartition inclusion:} There exists a block $\blockset_j \in \subpartition$ such that there is no block $\blockset^\star_{j^\star} \in \subpartition^\star$ such that $\blockset_j \subset \blockset^\star_{j^\star}$. As the subdivisions $\partition$ and $\partition^\star$ define two space of additive-per-block functions, we have for all $f \in E^{\subdivset^\star}_{\subpartition^\star}$, we have 
\[
f(\bx)=f_1(\bx_{\blockset^\star_1}) + \cdots + f_j(\bx_{\blockset^\star_j}) + \cdots + f(\bx_{\blockset^\star_B}),
\]
the family of elements in $E^{\subdivset^\star}_{\subpartition^\star}$ are derivable almost everywhere as product of almost everywhere derivable functions. Defining the differential operator $\frac{\partial^{|\blockset_j|}}{\partial \bx_{\blockset_j}}= \prod_{i\in \blockset_{j}} \frac{\partial}{\partial x_i}$, hypothesis give that $\frac{\partial^{|\blockset_j|}}{\partial \bx_{\blockset_j}}f=0$ for every $f\in E^{\subdivset^\star}_{\subpartition^\star}$. However the function $\phi : (x_1,\ldots, x_D) \mapsto \prod_{i \in \blockset_j} x_i$ is in $E^{\subdivset}_{\subpartition}$ and satisfy $\frac{\partial^{|\blockset_j|}}{\partial \bx_{\blockset_j}}\phi=1$ hence could not belong in $E^{\subdivset^\star}_{\subpartition^\star}$, this concludes the proof.
\end{proof}

\section{Block-predictors and their applications in the coastal flooding case}

Recall that our target function $y$ satisfies:
\[y(\bx)=y_1(\bx_{\blockset_1}) +\cdots + y_{\sizeblock}(x_{\blockset_\sizeblock}),\]
and that the constructed predictor is $\widehat{Y} = \bPhi^\top_1\widehat{\bxi}_1 + \cdots + \bPhi^\top_\sizeblock\widehat{\bxi}_\sizeblock$ (if the right subpartition has been found). Then, up to an additive constant (see Remark~\ref{rem:blockdecomp}), we have access to the block-predictors $\widehat{y}_i= \bPhi^\top_i\widehat{\bxi}_i$ of the block-functions $y_i$. The study of these block-predictors can bring a new light in the understanding of the impact of the variables over the target function.

\subsection{Results for the toy function}
For the 6D toy example in Section~\ref{subsec:MaxMod}, we can compare the results obtained from MaxMod with the target function $y$ in \eqref{eq:toolexample2Dblockadditiv}. Since $y$ is a sum of 2-dimensional block-functions, we can visualize the block-functions and their predictors for each $j = 1, \ldots, D/2$ using 3-dimensional plots. %The predictors are constructed using the MaxMod framework. The design of experiment is generated from a maximin LHD of size $4\times D_0$ with $D_0 = 6$. 
To be able to compare the block-functions $y_i$ with the block-predictors $\widehat{y}_i$, we plot the centered versions of these functions: $y_{c,i} = y_i-\int y_i$ and $\widehat{y}_{c,i} =  \widehat{y}_i -\int \widehat{y}_i$. After 17 iterations of MaxMod, the resulting predictor is defined in a finite-dimensional space of size $39$. The results, shown in Figure~\ref{fig:2Dplottoyfun}, are visually accurate, despite the predictors being piecewise linear approximations of the ground truth functions (see, e.g., the predictor of the $\arctan$ function). %For the third row, the predictor for the $\arctan$ function appears to be limited by the small dimensionality of the space. This limitation arises from the few number of observations and the stopping criterion based on the SE.
\begin{figure}[t!]
    \centering
    \includegraphics[width=0.45\linewidth]{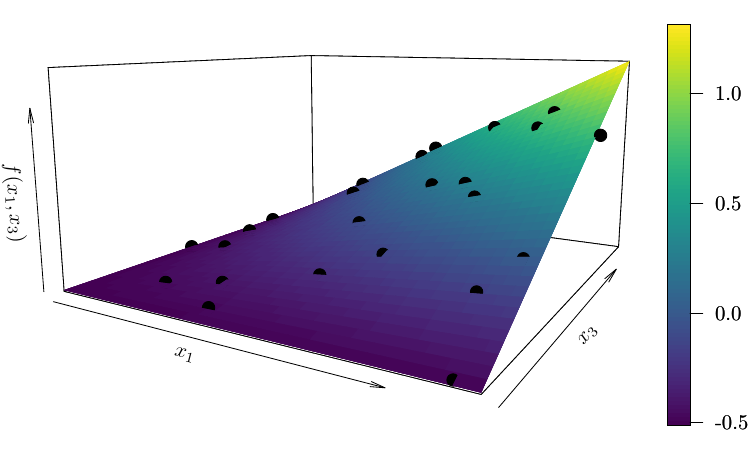}
    \includegraphics[width=0.45\linewidth]{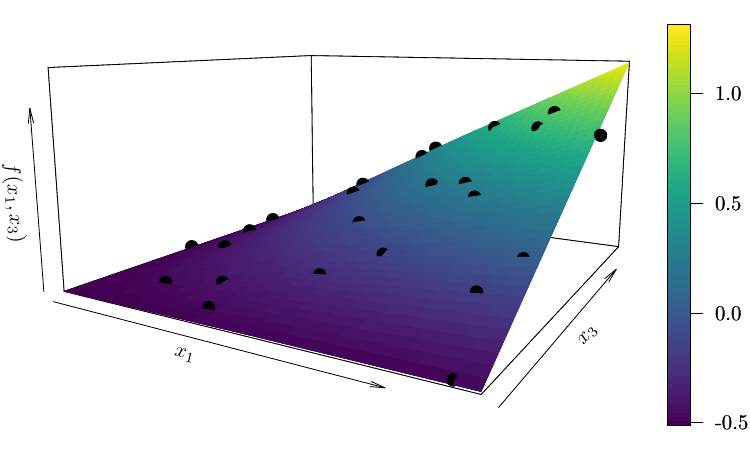}
    \includegraphics[width=0.45\linewidth]{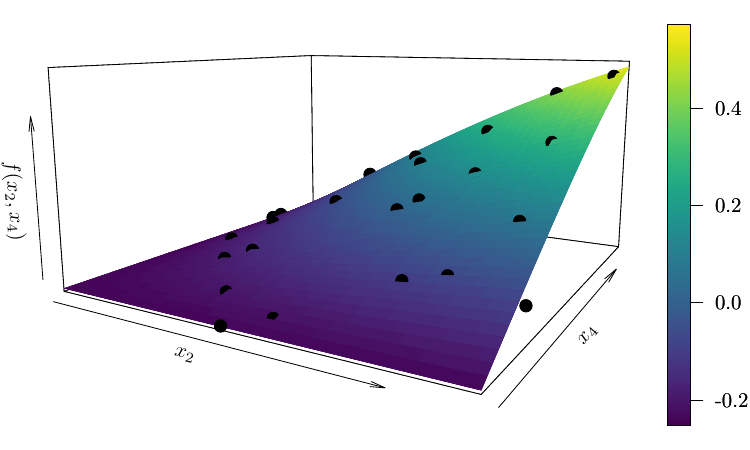}
    \includegraphics[width=0.45\linewidth]{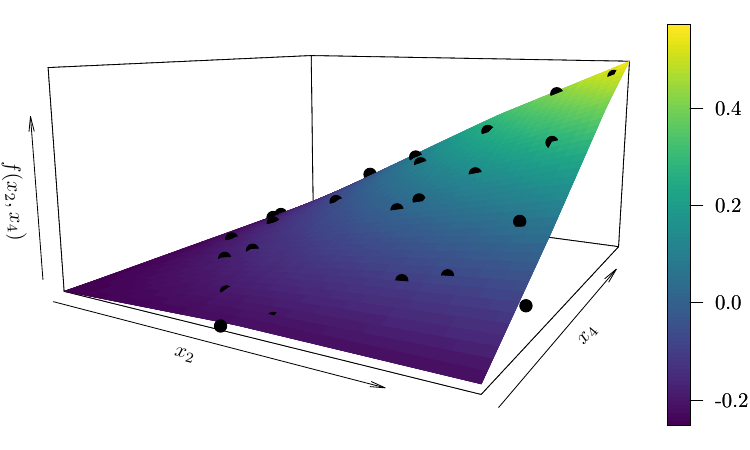}
    \includegraphics[width=0.45\linewidth]{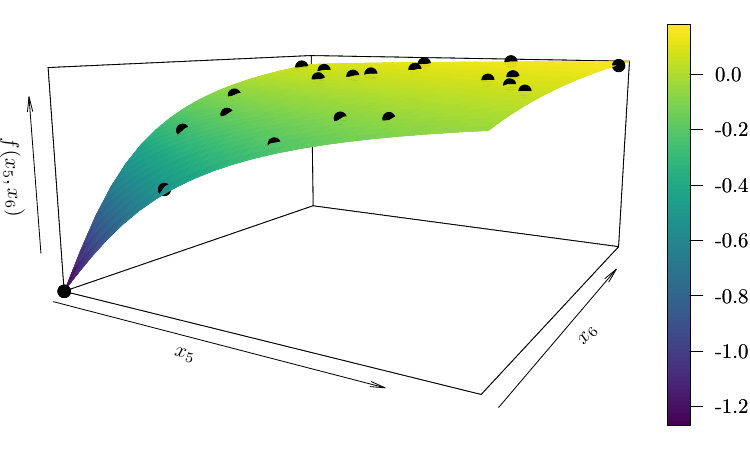}
    \includegraphics[width=0.45\linewidth]{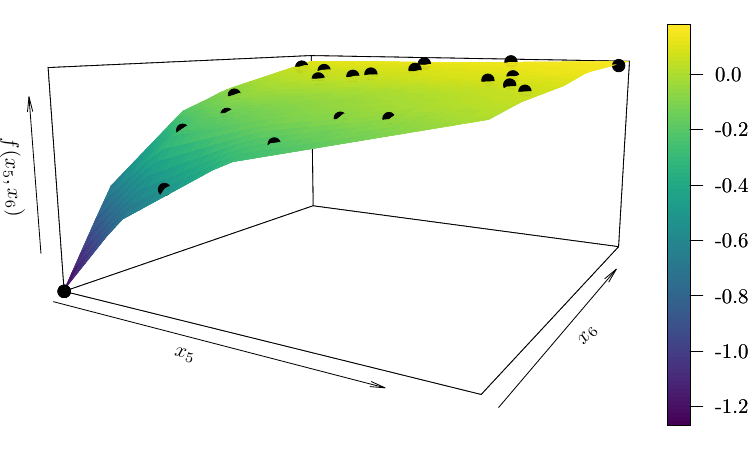}
    \caption{2D visualizations of the centered functions (top) $y_1:(x_1,x_3) \mapsto 2x_1x_3$, (middle) $y_2:(x_2,x_4) \mapsto \sin(x_2x_4)$ and (bottom) $y_3:(x_5,x_6) \mapsto \tan(3x_5+5x_6)$. The ground truth functions and their corresponding predictors are shown in the left and right panels, respectively.}
    \label{fig:2Dplottoyfun}
\end{figure}

\subsection{Analysis for the coastal flooding application}
\label{appendix:Analysis for the coastal flooding application}

As discussed in Section~\ref{subsec:coastalAPP}, %, and observed in Figures~\ref{fig:surgeTide_graphANOVA}, \ref{fig:BRGM_2dplot} and \ref{fig:MaxModEnergyBRGM}, 
the inferred additive structure of the target function $y := \log_{10}(A_{flood})$ is given by $\widehat{y}(S,T,\phi, t_{+},t_{-})= \widehat{y}_1(S,T,\phi) + \widehat{y}_2(t_{+}) + \widehat{y}_3(t_{-})$. %We here study the block-predictor of the function $f_1$ by visualizing bivariate representations of the functions $f_1$ for some fix values of $\phi$. 
Figure~\ref{fig:2Dplotcostalfun2} illustrates that, for a fixed $\phi$, the function $(T, S, \phi) \mapsto \widehat{y}_1(T, S, \phi)$ is quasi-linear. Specifically, the contour lines for small values of $S$ and $T$ are evenly spaced straight lines, indicating that $\widehat{y}_1(\cdot,\cdot,\phi)$ approximately behaves as a linear function. However, non-linear interactions are observed only for high values of $T$ and $S$. Independently of the value of $\phi$, the vertical orientation of the contour lines highlights that the variable $T$ has a greater influence on coastal flooding than $S$. This is consistent with the Sobol analysis shown in Figure~\ref{fig:surgeTide_graphANOVA}. It can also be observed that the influence of the tide $T$ on coastal flooding increases as $\phi$ decreases. This suggests that coastal flooding is more sensitive to the tide when it is synchronized with the surge. Conversely, the influence of the surge peak $S$ on coastal flooding does not appear to increase as $\phi$ decreases. These observations are intuitive, given that the range of the tide $T$ is broader than that of the surge $S$ prior to renormalization, as shown in Figure~\ref{fig:surgeTide_graphANOVA}.

\begin{figure}
\centering
\includegraphics[width=0.45\linewidth]{figs/BRGM_2dplot1.pdf}
\includegraphics[width=0.45\linewidth]{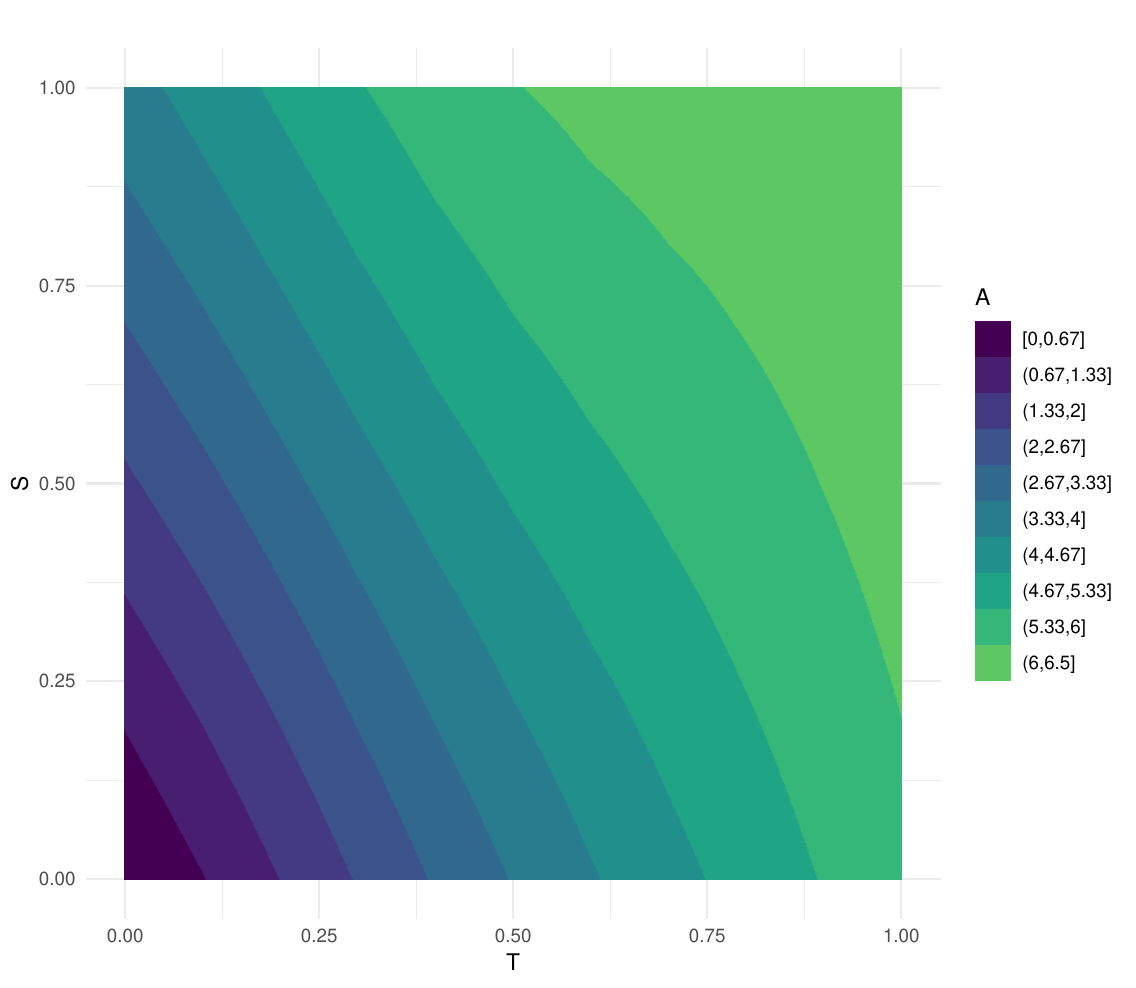}

\includegraphics[width=0.45\linewidth]{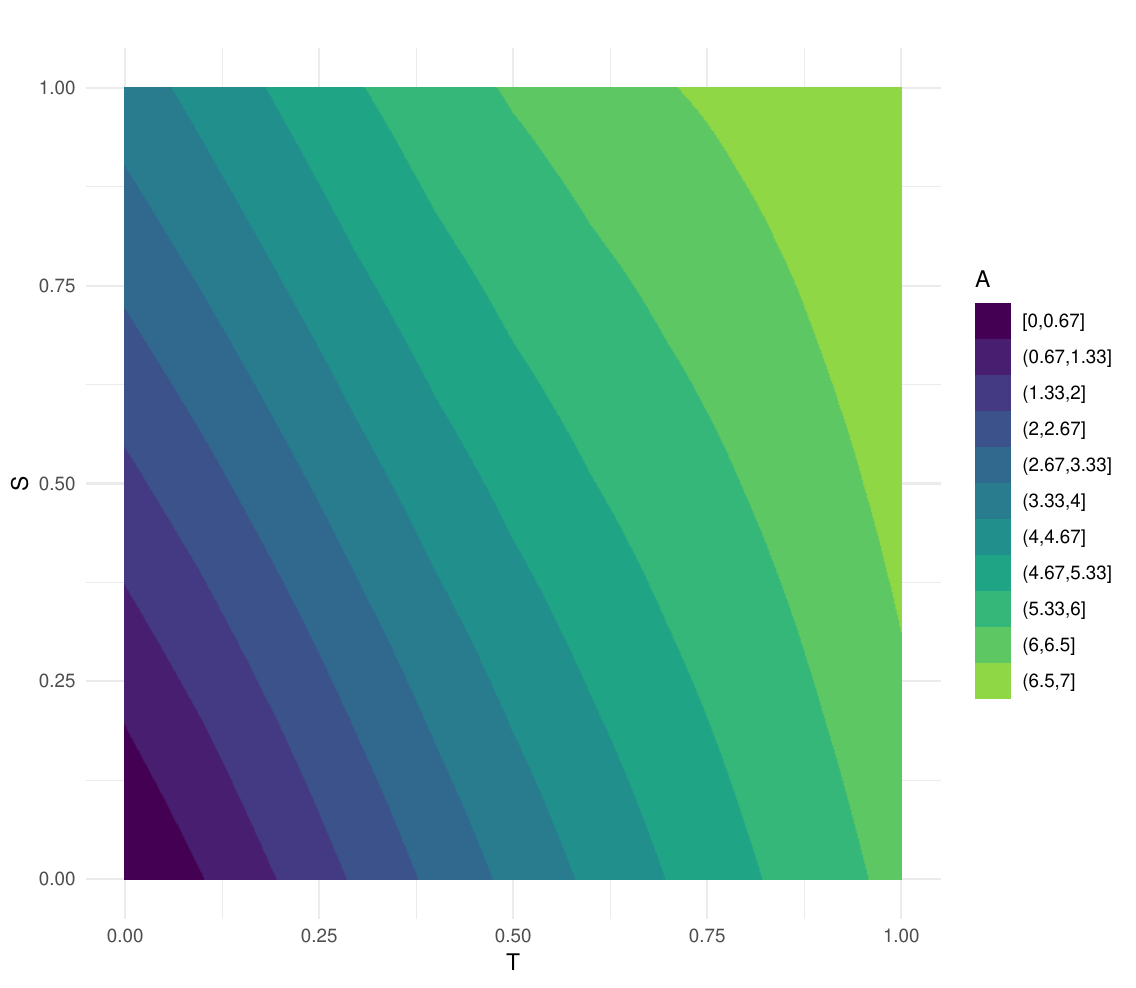}
\includegraphics[width=0.45\linewidth]{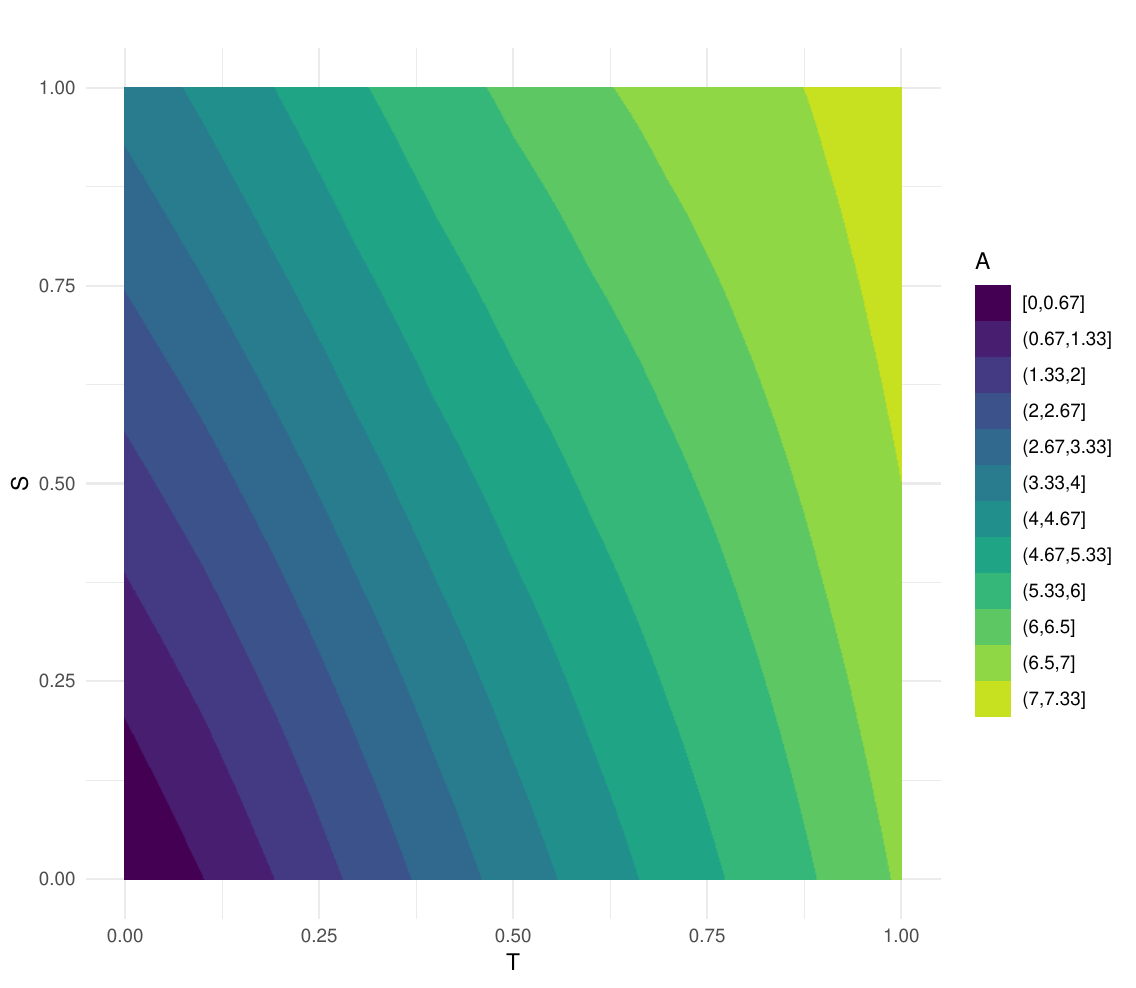}

\includegraphics[width=0.45\linewidth]{figs/BRGM_2dplot5.pdf}
\caption{
%Bivariate representation of the functional relationship between $S$ and $T$ given $(\cos(\phi)+1)/2=0,1/4$ (top), $(\cos(\phi)+1)/2=1/2,3/4$ (middle), $(\cos(\phi)+1)/2=1$ (bottom).
Bivariate representation of $\widehat{y}_1(S, T, \phi)$ for $\displaystyle \phi = \pi, \frac{2\pi}{3}, \frac{\pi}{2}, \frac{\pi}{3}, 0$, presented in order of appearance).
}
\label{fig:2Dplotcostalfun2}
\end{figure}

\end{document}